\documentclass[10pt]{sig-alternate-10pt}
\usepackage{times}
\usepackage{balance}
\usepackage{url}
\usepackage{epsfig}
\usepackage{graphicx}
\usepackage{subfigure}
\usepackage{amsmath}
\usepackage{amssymb} 
\usepackage{amsfonts}
\usepackage{color}
\usepackage{multirow}
\usepackage{algorithm}
\usepackage[utf8]{inputenc}
\usepackage[T1]{fontenc}

\newcommand{\para}[1]{{\vspace{4pt} \bf \noindent #1 \hspace{10pt}}}

\begin{document}

\title{Revisiting Degree Distribution Models for \\Social Graph Analysis}
\author{Alessandra Sala$^\S$, Sabrina Gaito$^\dag$, Gian Paolo Rossi$^\ddag$, Haitao Zheng$^\S$ and Ben Y. Zhao$^\S$ \\
{$^\S$UC Santa Barbara, $^\dag$Universit\`{a} degli Studi di Milano}\\
{\rm \em \{alessandra, htzheng, ravenben\}@cs.ucsb.edu, gaito@dsi.unimi.it,
  rossi@dico.unimi.it}}

\newtheorem{theorem}{Theorem}
\newtheorem{definition}{Definition}
\newtheorem{lemma}{Lemma}
\newtheorem{cor}{Corollary}
\newtheorem{fact}{Fact}
\newtheorem{property}{Property}
\newtheorem{remark}{Remark}
\newtheorem{claim}{Claim}

\maketitle

\begin{abstract}
  Degree distribution models are incredibly important tools for analyzing and
  understanding the structure and formation of social networks, and can help
  guide the design of efficient graph algorithms.  In particular, the
  Power-law degree distribution has long been used to model the structure of
  online social networks, and is the basis for algorithms and heuristics in
  graph applications such as influence maximization and social search.  

  Along with recent measurement results, our interest in this topic was sparked by our
  own experimental results on social graphs that deviated significantly from
  those predicted by a Power-law model.  In this work, we seek a deeper understanding of
  these deviations, and propose an alternative model with significant
  implications on graph algorithms and applications.  We start by quantifying
  this artifact using a variety of real social graphs, and show that their
  structures cannot be accurately modeled using elementary distributions
  including the Power-law.  Instead, we propose the Pareto-Lognormal (PLN)
  model, verify its goodness-of-fit using graphical and statistical methods,
  and present an analytical study of its asymptotical differences with the
  Power-law.  To demonstrate the quantitative benefits of the PLN model, we
  compare the results of three wide-ranging graph applications on
  real social graphs against those on synthetic graphs generated using the
  PLN and Power-law models.  We show that synthetic graphs generated using
  PLN are much better predictors of degree distributions in real graphs, and
  produce experimental results with errors that are orders-of-magnitude smaller than
  those produced by the Power-law model.
\end{abstract}

\section{Introduction}
\label{sec:intro}

Graph degree distributions are fundamental tools used in the study of complex
networks such as online social networks.  Not only do they reveal insights
into the structure and formation of these networks, but they also lay the
foundation for modeling network dynamics and help guide the design of
graph algorithms and applications.  In particular, it is widely believed that the
Power-law distribution accurately captures node degrees in complex networks
such as online social networks, {\em i.e.} their degree distribution follows
a Power-law function $f(x)=c x^{-\alpha}$, for some normalization constant
$c$ and exponent $\alpha$. As a result, the Power-law model has already
played a significant role in guiding the design of algorithms on social
network problems such as influence maximization~\cite{influence_kdd09},
landmark selection~\cite{social_landmark}, and link privacy
protection~\cite{link_privacy}.

\para{Deviating from Power-law.} Our interest in verifying the Power-law
model for social graphs was first sparked by our own efforts to partition social
graphs. In a recent project, we searched for an efficient way to partition and
divide large social graphs across distributed machines for parallel graph
computation.  The goal is to minimize edges between partitions, thereby
reducing data dependencies between machines when resolving graph queries.

Given the known density of social graphs, it was not surprising when popular
partitioning algorithms such as Metis produced poor partitions (those with
low modularity) from our Facebook graphs~\cite{interaction}.  Our next step
was to leverage the known fact that Power-law graphs are vulnerable to
targeted attacks, {\em i.e.} they quickly fragment when their ``supernodes''
are removed~\cite{powerlawfailure}.  We evaluated a new partitioning approach
where a small portion of nodes with the highest degree are replicated across
all machines. This allows us to essentially ``remove'' the supernodes from
the graph, which should fragment the graph into numerous disconnected or
weakly connected subgraphs, which are easily partitioned.

Surprisingly, our results showed that unlike prior results on peer-to-peer
networks~\cite{gnutella02}, social graphs were extremely resilient to this
approach.  On one of our typical Facebook graphs with 500K nodes, cutting the
top 10\% (50K) supernodes had no impact on the connected graph. Nearly all of
the remaining nodes ($\sim$440K nodes) were still connected in a strongly
connected component.

\para{Revisiting degree models.}  These surprising results led us to
re-examine how well social graphs obey the Power-law degree model.  Recent
measurements of popular online social networks have shed light on their
internal structure, and also produced a number of real social graphs suitable
for model validation~\cite{interaction,
  socialnets-measurement,Mislove09_wosn,flickr-growth}.  Despite using the
Power-law parameter as a popular graph metric, it is increasingly evident
that the Power-law only captured a portion of the degree distribution
curve.  While never explained, several results consistently show that the
Power-law distribution over-predicts the number of ``high-degree nodes'' in
social networks~\cite{Ahn07, sampling_facebook}.  Similar divergence has
also been observed in Internet topology and web graphs or mobile cell
graphs~\cite{winnersdon't, mobile_dpln}.

In this paper, we revisit the problem of modeling node degree distributions
in online social networks.  We question existing assumptions, propose the use
of an alternative distribution model, and show that choosing an appropriate degree
distribution model 
has dramatic impact on a wide-range of OSN and graph applications.  We are
primarily interested in answer three key questions.  First, how significant
is the fitting error when modeling degree distributions using the Power-law
and other elementary distributions?  Second, can we propose an alternative
distribution that provides a fit with significantly better accuracy?
Finally, what is the real impact of switching to this alternative
distribution from Power-law, and how does it impact real graph applications?

\para{Results and Contributions.}  In addressing the questions outlined
above, our work makes several contributions to the problem of developing a
more accurate degree distribution model for online social networks.

{\em First}, we assert that the seminal Power-law distribution does
not accurately model degree distributions in OSNs.  We verify this hypothesis
using a number of measured social graphs from the Facebook and Orkut OSNs,
and show that other elementary distribution also perform poorly.  {\em Second}, we
search for a more accurate model by examining complex distributions, and
propose the use of the Pareto-Lognormal (PLN) distribution.  Intuitively, the
Power-law distribution forms when each new node joins the network via a
bootstrap node $b_i$, and builds edges using the ``rich-get-richer''
model in some graph neighborhood around $b_i$.  In contrast, the PLN
distribution forms when each new node joins multiple communities and uses a
stochastic process to drive connections within each community.  To support
PLN's use in analysis, we analytically describe its probability distribution
function, cumulative distribution function, and maximum likelihood function.
Using three different error measures, we compare the accuracy of
PLN, Power-law and 5 other models on 7 different OSN social graphs ranging
in size from 740K edges and 14K nodes, to 118 million edges and 1.6 million
nodes.  Our results confirm that PLN provides the most accurate edge
distribution model.

{\em Third}, we highlight the differences between the PLN and
Power-law models, quantifying the asymptotic differences between them when predicting 
high degree nodes in the network.  For both models, we derive a 
close form to bound the lowest degree of a node for any percentile of
high-degree nodes in the network. We then make predictions of the cardinality
of nodes in that percentile.  Using our social graph data, we validate the
predictions from both models, and find that PLN generally produces errors
that are at least two-orders-of-magnitude smaller than those of Power-law.

{\em Fourth}, we examine the end-to-end impact of degree distribution models
on multiple graph applications, including graph partitioning, influence
maximization~\cite{influence_kdd09}, and attacks on social graph link
privacy~\cite{link_privacy}.  Prior studies of these applications used assumptions of
Power-law graphs to drive their algorithm design.  We implement each
application, and show that running them on synthetic Power-law
graphs produced {\em dramatically} different results from those on real social graphs .  In
contrast, we show that running applications on synthetic PLN graphs produces
results nearly identical to those produced using real social graphs.  {\em
  Finally}, we give some preliminary intuition towards the ongoing design of
a generative model that both captures the temporal evolution of social
networks and produces degree distributions matching our PLN model.  We draw
insights from a series of daily snapshots of a Facebook social graph that
capture dynamic growth over a period of a month.

\para{Social Graph Datasets.} To evaluate our models, we use 7
real social graphs gathered from recent measurements of popular online social
networks.  The majority of our datasets come from Facebook, the most popular
OSN today with more than 500 million users.  We use traces gathered through
crawls of the Facebook network in $2009$, when Facebook was structurally
organized into geographical/regional networks.  We had crawled and
analyzed over $10$ million users ($\sim$15\% of Facebook in 2008), as part
of an earlier measurement study~\cite{interaction}.  For
this paper, we utilize $6$ anonymized social graphs representing a range of
network sizes, from a small Monterey, CA graph ($13K$ nodes, $704K$ edges) to
a large London graph ($1.6$ million nodes, $118$ million edges).  We
also include in our study the ``Manhattan Random Walk'' Facebook graph
from~\cite{sampling_facebook}, which has been proven to be a representative
uniform random sampling of the total Facebook graph.  We use 
this graph to validate the representativeness of our Facebook results.  Finally, we
also include a public social graph from the Orkut
OSN~\cite{socialnets-measurement}.  With 3 million nodes, 111 million edges,
it has more nodes but less edges than our largest Facebook graph (London).
Our datasets and their key statistics are summarized in Table~\ref{table:stats}.

\para{Roadmap.}  We begin in Section~\ref{sec:elementary} by examining how
well elementary functions fit degree distributions from real social
graphs. Next, in Section~\ref{sec:newdistribution}, we describe the PLN model
through its PDF, CDF and maximum likelihood functions, and show the accuracy of this model
through statistical analysis.  Then, we quantify 
asymptotic differences between the models in Section~\ref{sec:implication},
present experimental application-level results in
Section~\ref{sec:application}, and discuss intuition for a generative model
in Section~\ref{sec:generative}.  Finally, we discuss related work in
Section~\ref{sec:related} and conclude in Section~\ref{sec:conclusion}.
\section{Elementary Distributions }
\label{sec:elementary}
We start by examining how well elementary distribution models
fit real social graphs from deployed OSNs. We include in our analysis the three
elementary distributions:  Power-law, Lognormal and Exponential
distributions~\cite{mitzenmacher03, powerlawfit}.
We leave out the formal introduction of these models, and instead focus on presenting
the results of our experimental analysis.

\begin{figure*}[t]
\begin{minipage}[h]{0.66\columnwidth}
\epsfig{figure=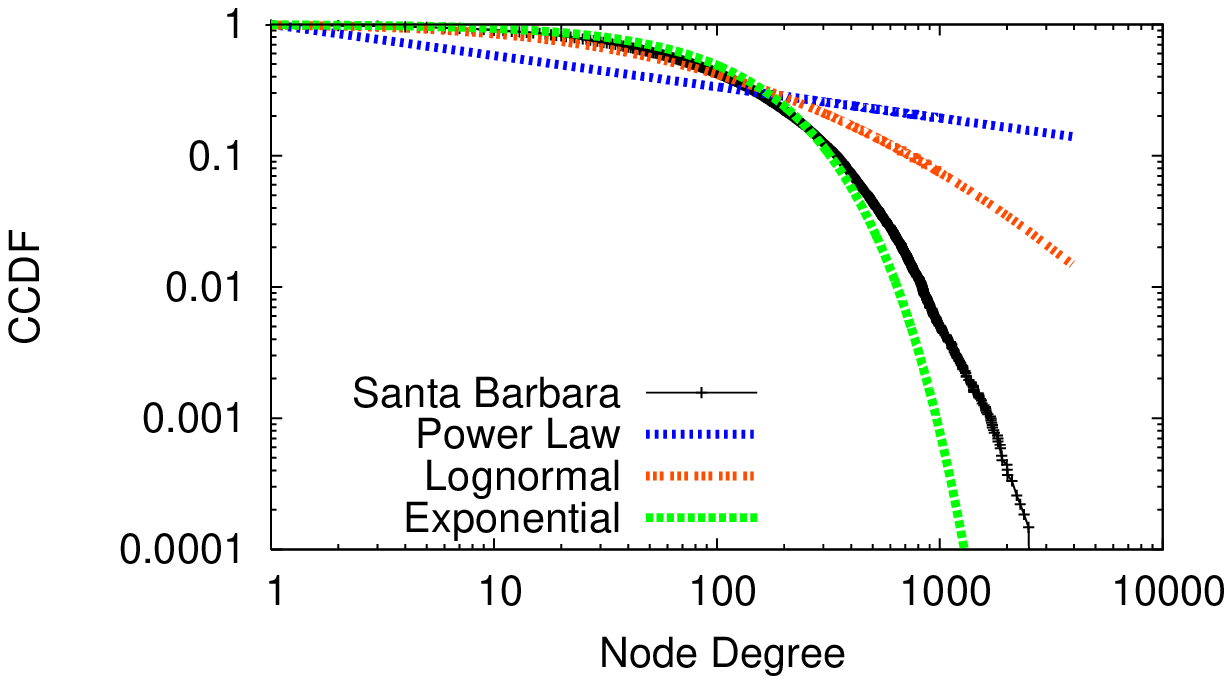, height=1.25in}
\end{minipage}
\hfill
\begin{minipage}[h]{0.66\columnwidth}
\epsfig{figure=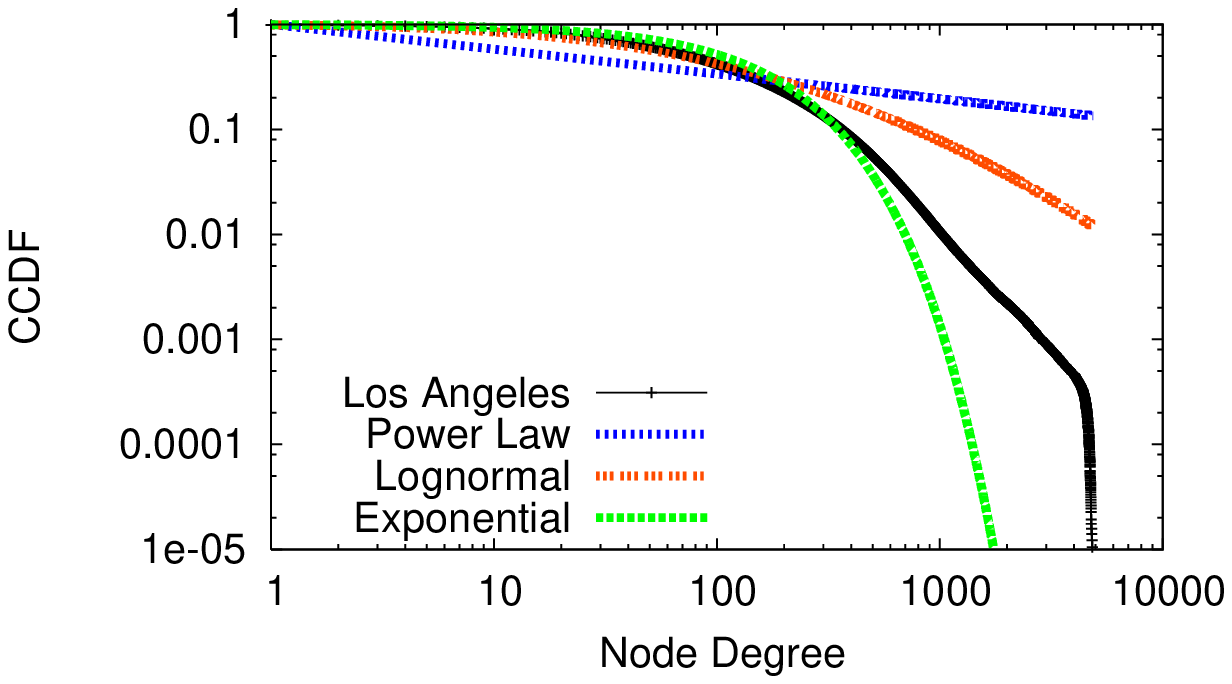,height=1.25in}
\end{minipage}
\hfill
\begin{minipage}[h]{0.66\columnwidth}
\epsfig{figure=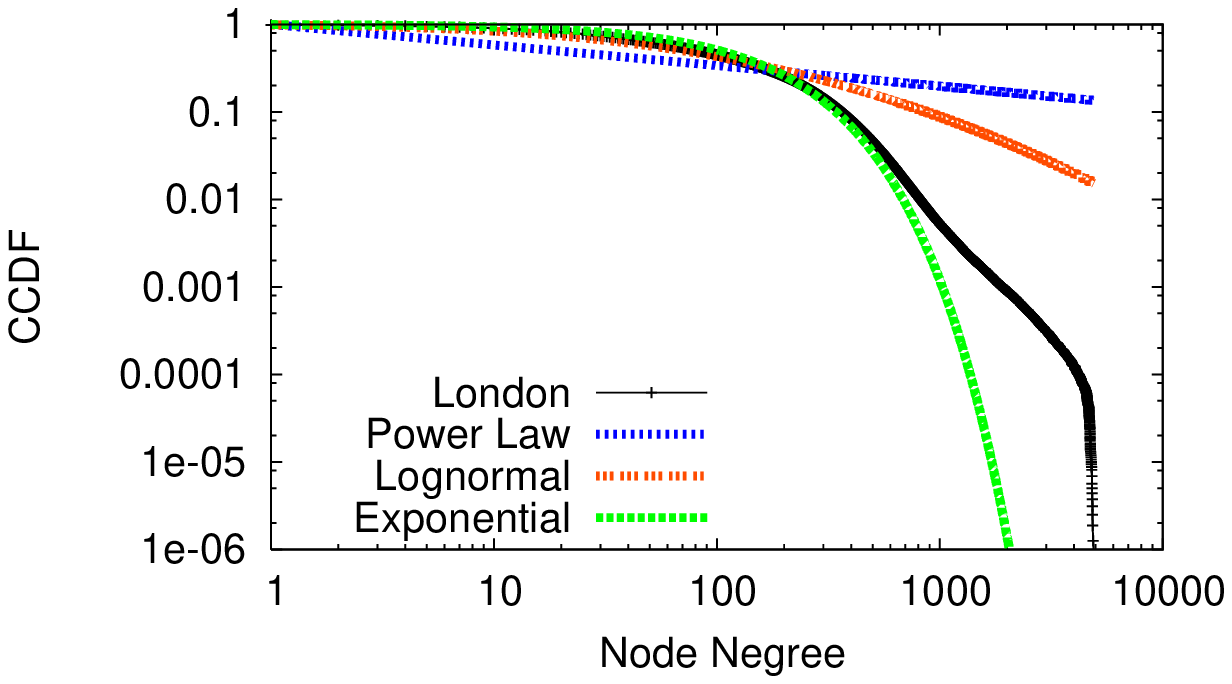,height=1.25in}
\end{minipage}
\caption{\small Complementary cumulative distribution function of elementary models
  fitted on Santa Barbara, Los Angeles and London Facebook graphs.} 
\label{fig:ComplementaryCDF}
\begin{minipage}[h]{0.66\columnwidth}
\epsfig{figure=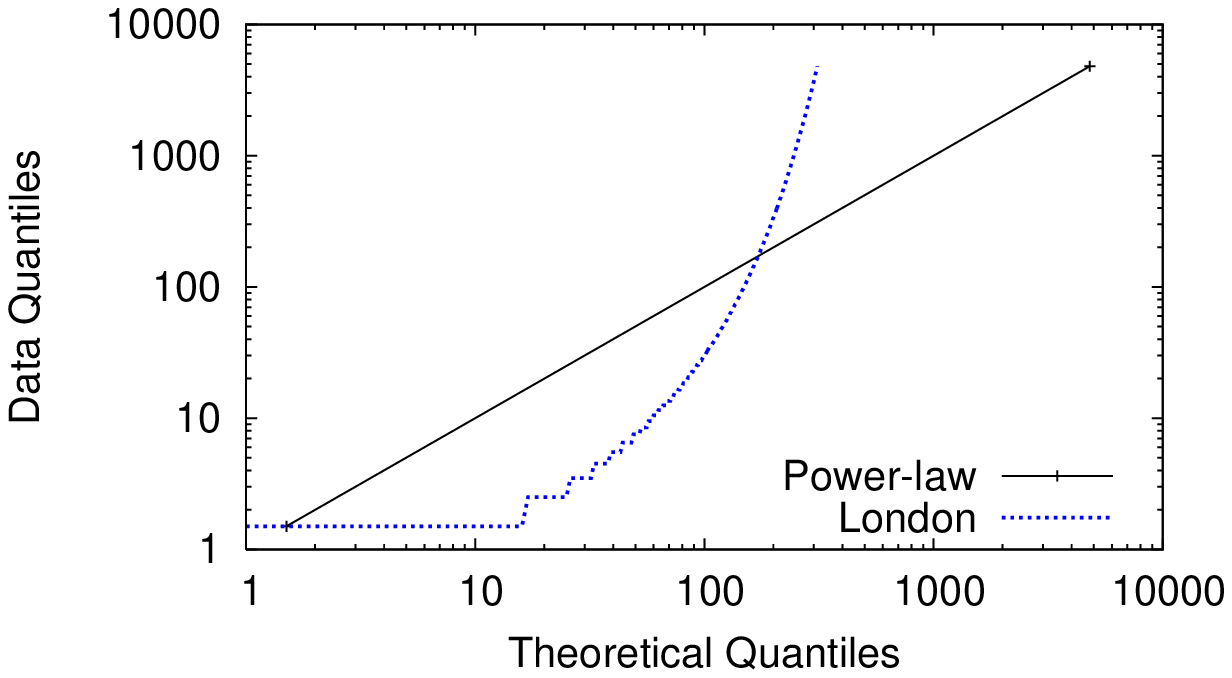, height=1.25in}
\end{minipage}
\hfill
\begin{minipage}[h]{0.66\columnwidth}
\epsfig{figure=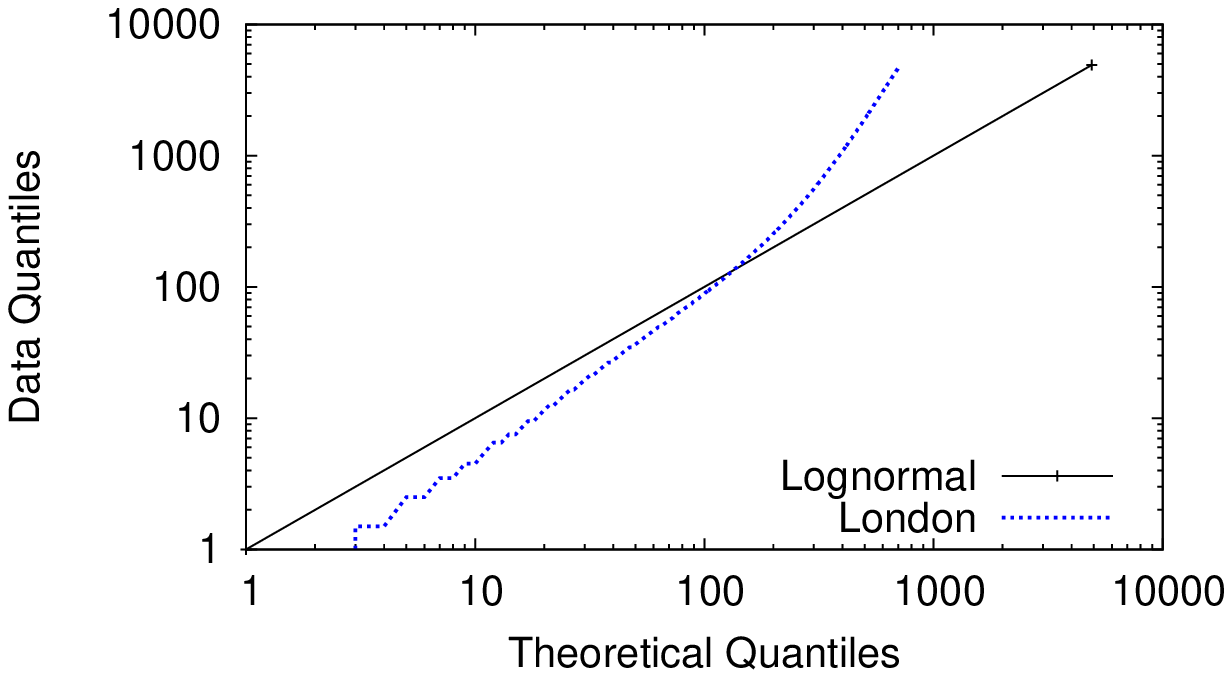,height=1.25in}
\end{minipage}
\hfill
\begin{minipage}[h]{0.66\columnwidth}
\epsfig{figure=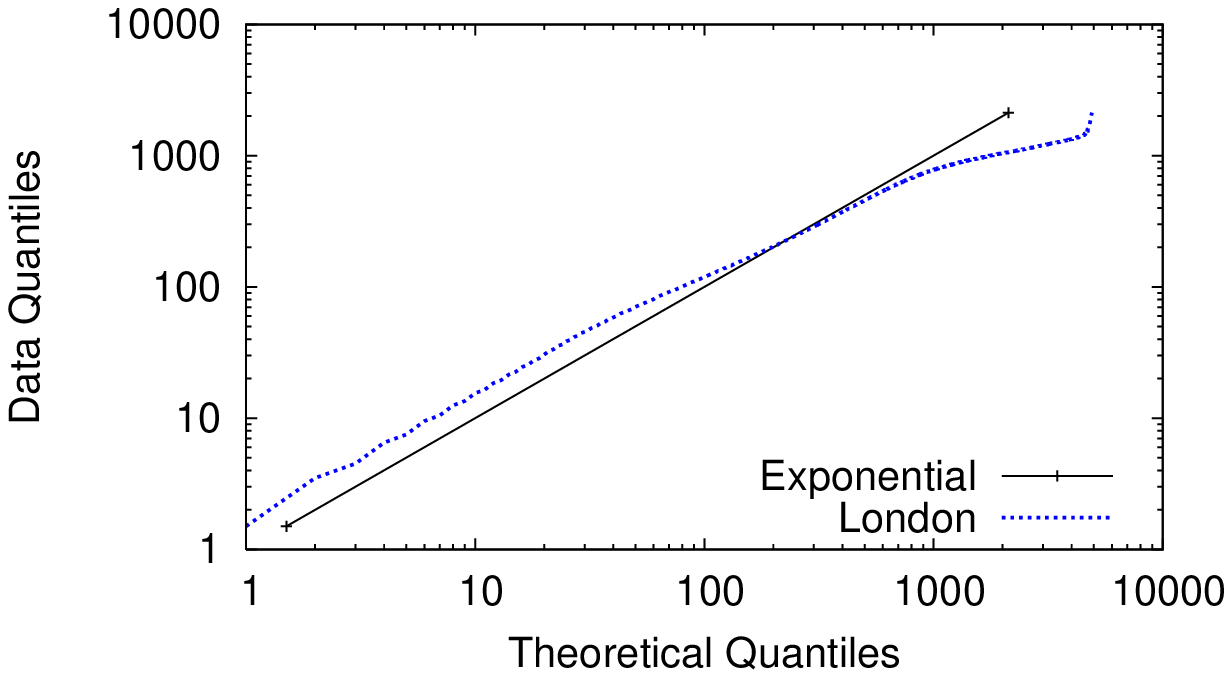,height=1.25in}
\end{minipage}
\caption{\small Quantile-Quantile plot of elementary distributions fitted on the London dataset to graphically show the similarity of real social network data to a particular distribution model.}
\label{fig:quantile}
\end{figure*}

\begin{figure*}[t]
\begin{minipage}[h]{0.66\columnwidth}
\centering
\epsfig{figure=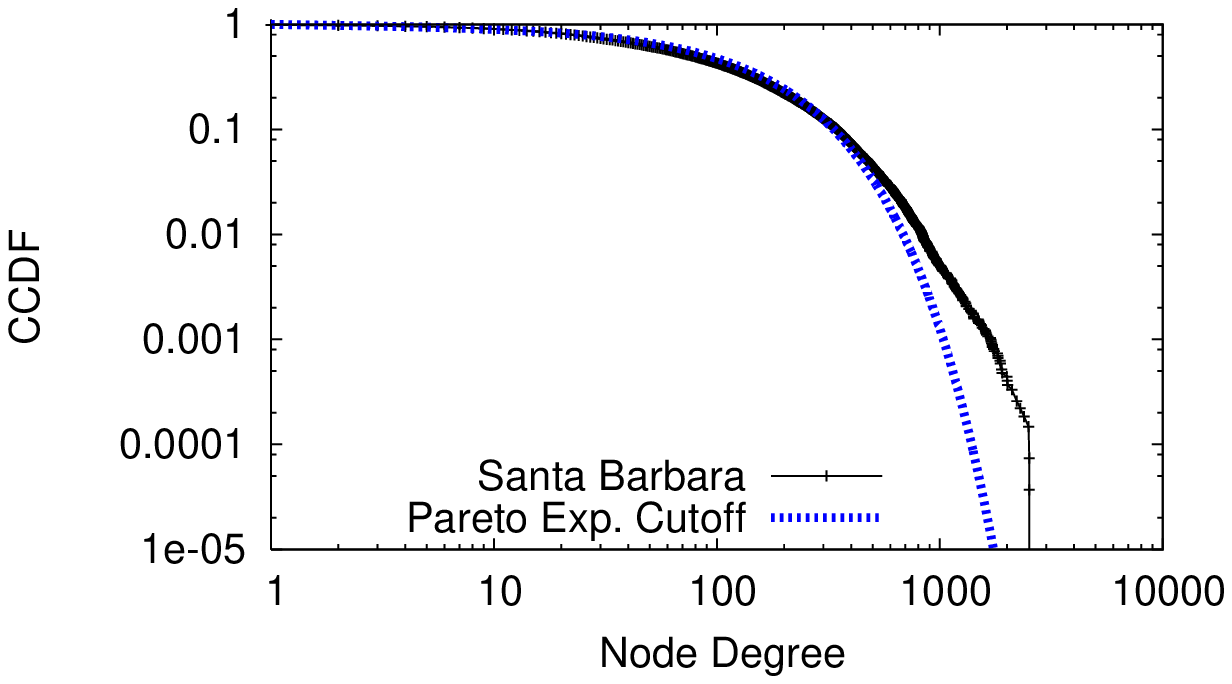,height=1.25in}
\end{minipage}
\hfill
\begin{minipage}[h]{0.66\columnwidth}
\centering
\epsfig{figure=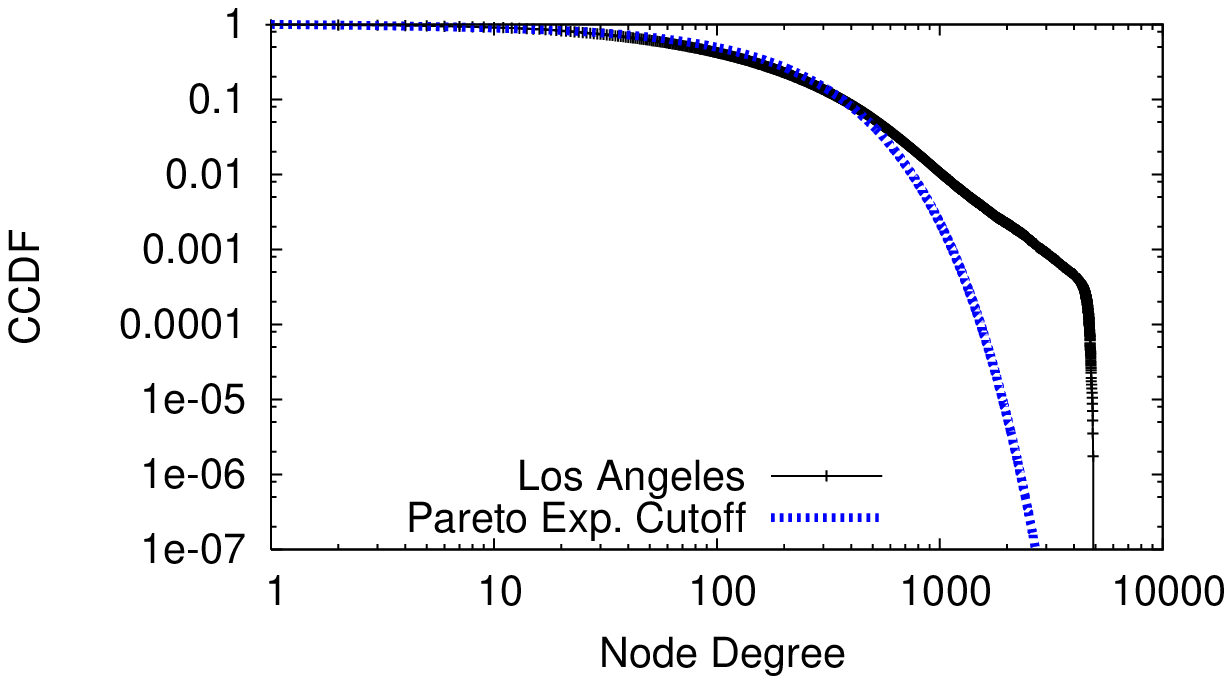,height=1.25in}
\end{minipage}
\hfill
\begin{minipage}[h]{0.66\columnwidth}
\centering
\epsfig{figure=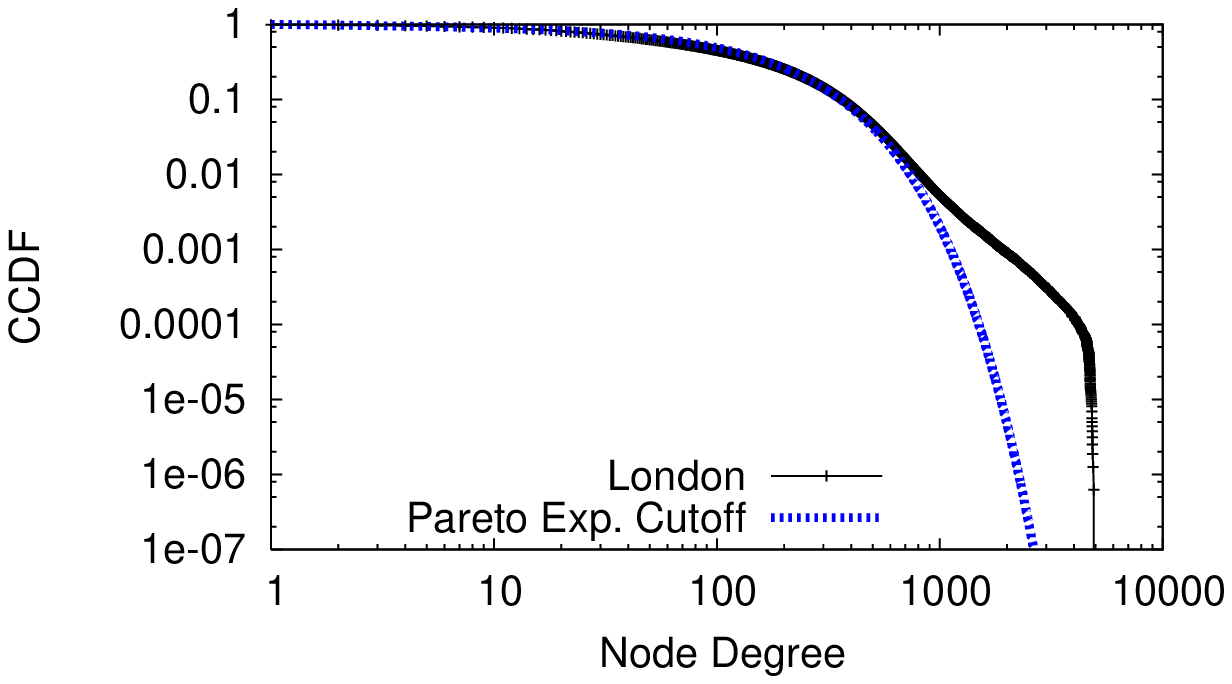,height=1.25in}
\end{minipage}
\caption{\small The CCDF of the Pareto with Exponential Cutoff fitted on Santa
  Barbara, Los Angeles and London networks.} 
 \label{fig:ccdf_cutoff}
\end{figure*}

\para{Fitting Models to Real Datasets.}  
We fit each distribution model to our real OSN
datasets, using an optimal estimator to derive the best model
parameters~\cite{powerlawfit}.

For each considered distribution model, we analyze the quality of the fit on
our real data through three probability plots: the probability distribution, the
complementary cumulative distribution~(CCDF), and the ``quantile-quantile''
plot.

The {\em probability distribution function} is an important metric for
understanding the portion of nodes that have a certain degree value.  For
example, the seminal Power-law model predicts that a large number of nodes
have a high degree.  The CCDF quantifies how often a node's degree is above a
value. It is particularly useful for identifying the general slope of OSN
connectivity.  Finally, the {\em Quantile-Quantile} analysis can graphically
compare different distributions, in our case a theoretical model and a real
social graph.  The plot shows the discrepancy between degree values that
correspond to the same quantile in both distributions.  The greater the
distance from the reference line, the stronger the evidence that the dataset
follows a different distribution.

\para{Experimental Results.}
We perform experiments on our $6$ Facebook datasets, the random walk Facebook
graph, and the Orkut dataset.  Results across all graphs are highly
consistent, and we only show 3 graphs for brevity, Santa Barbara Los Angeles and London
Facebook networks. These social graphs range in size from $13$K to
$1.6$M nodes.

Figure~\ref{fig:ComplementaryCDF} plots the CCDF of the fitted models on the
three datasets.  The CCDF presents a detailed view on the tail of these
distributions. We see that none of the presented models accurately captures
the decay slope shown by the real datasets.  Both the Power-law and the Lognormal
distributions overestimate the tail. Take London for example
(Figure~\ref{fig:ComplementaryCDF}). The Power-law model overestimates 
the highest degree nodes by up to $5$ orders of magnitude compared to the real
graph.  In contrast, the Exponential distribution decays sharply, and significantly
underestimates the number of high degree nodes. On the London graph, this
error reaches $4$ orders of magnitude for nodes with degrees of $2000$.
 
The quantile-quantile plot formally quantifies the distance between the real
data and the model.  Figure~\ref{fig:quantile} shows that both the Power-law and the
Lognormal model underestimate nearly $90\%$ of nodes ({\em i.e.}
Figure~\ref{fig:quantile}, $x$-axis $\in [1, 500]$) and largely overestimate
the very high degree nodes. 
 In contrast,
the Exponential distribution model exhibits a reasonably good behavior along
the lower tail and the body of the distribution compared to the dataset,
although a strongly diverging slope is displayed on the high degree nodes.

We do not include plots on the comparison of probability distribution
functions because of space constraints. However, they lead to the same conclusion:
 the high density of
nodes with low degree are not well predicted by the Power-law and
Lognormal models, and all models fail to accurately predict the number of
high degree nodes.  Finally, while this section only reports results from
graphical assessment of goodness-of-fit, we confirm our observations using 
statistical error measures later
in Section~\ref{sec:statistical_analysis}.

 \section{Complex distribution models}\label{sec:newdistribution}

Our experimental results show that none of the elementary distributions, including
the Power-law distribution, provide a satisfactory fit for degree
distributions in today's OSNs.  These results, however, do lead to an
intuition where low degrees are distributed following a Pareto model while
higher degrees can be modeled with an asymptotically faster
decaying distribution.  We believe a combination of two models will offer a
better representation of the complex phenomena observed on the OSNs.  Similar
approach has been successfully used by other fields, for example, modeling
actuarial data in economics~\cite{Weibull_LNP}.  

In this section, we test our intuitions by evaluating four distribution
models with a Pareto component, each with varying degrees of accuracy and
fitting complexity.  Ultimately, we confirm via both graphical and statistical
analysis that the beginning of these distributions is characterizable
as a Pareto and the upper tail decays as a Lognormal, and that the
Pareto-Lognormal distribution provides the best combination of accuracy and
fitting overhead.

\subsection{Pareto with Exponential Cutoff}
Prior work has proposed modeling human behavior using the Pareto with
Exponential cutoff distribution~\cite{pareto_cutoff}, where
the beginning of the distribution follows a Pareto and the tail exponentially
decays. This formulation is particularly interesting to our study 
given the decay observed in the tail of our data (see Figure~\ref{fig:ComplementaryCDF}). 
We study this distribution using the derivations from~\cite{parto_exp_cutoff}.

\para{Fitting to real datasets.}  Figure~\ref{fig:ccdf_cutoff} shows how well
the Pareto with Exponential cutoff distribution fit the Santa Barbara, Los Angeles and
London datasets.  We plot the CCDF to graphically show how the tail of this
distribution decays compared to real datasets.  While this model is able to
closely fit the low degree nodes of the real data, it fails short in capturing
the upper tail by underestimating the density of nodes with high degrees.  We
conclude that the exponential cutoff is too sharp to properly capture OSN
connectivity of central (high-degree) nodes.  Intuitively, we are looking for 
a more gradual decay slope.

\subsection{Pareto-Lognormal Distribution}
\label{sec:likelihood_formalization}
Given the low accuracy in modeling the upper tail of our data with the Pareto
with exponential cutoff, 
we turn our attention to a family of distributions that mixes Pareto with
Lognormal distributions. These models, compared to the Pareto with exponential
cutoff model, provide a smother decay on the upper tail
that provides a better match for our real data.  
We start with the Double Pareto-Lognormal (DPLN)
distribution introduced by Reed in~\cite{Reed03thedouble}.  The DPLN is a
complex model with the ability to fuse two Pareto and a Lognormal
distribution.  It includes four parameters: two Pareto
exponents $\alpha$ and $\beta$ that identify the slope of the upper and
lower tails of the distribution, and $\mu$ and $\tau$ that describe the
Lognormal parameters connecting the two Pareto tails.  The DPLN also gives
rise to two other distributions~\cite{Reed03thedouble}: the Pareto-Lognormal (PLN) and the
Lognormal-Pareto (LNP).  Both are expressed by three
parameters: Pareto exponent $\beta$, and Lognormal
components $\mu$ and $\tau$. 

Next, we derive the precise formulation of the PLN distribution in terms of the
PDF, CDF and the likelihood function.  We omit the detailed derivation on
DPLN and LNP for brevity. In Section~\ref{sec:likelihood_formalization} and
\ref{sec:statistical_analysis}, we use both graphical and statistical assessments
of goodness-of-fit to compare these three distributions.

\para{Pareto-Lognormal PDF and CDF Derivation.} The PLN is expressed as the
combination of two probability distributions.  

We derived the correct formula of this distribution (which is a {\em
  limit form} of the DPLN) using results on the
DPLN~\cite{Reed03thedouble}. 
The Pareto-Lognormal probability distribution function is:
\begin{equation}\label{pdf_pareto_log}
f(x)=\beta x^{\beta-1}e^{(-\beta\mu+\frac{\beta^2\tau^2}{2})}\Phi^c\left( \frac{\log{x}-\mu+\beta \tau^2 }{\tau}\right)
\end{equation}
where $\beta$ characterizes the slope of the lower 
tail of this distribution which follows a Pareto behavior,
and $\mu$ and $\tau$ characterize the body and the 
upper tail of this distribution, which approximate a Lognormal 
decline. We also derived the Pareto-Lognormal cumulative distribution
function, formalized as follows:

\begin{small}
\begin{equation}\label{cdf_pareto_log}
 F(x)=\Phi\left( \frac{\log{x}-\mu}{\tau}\right)+x^{\beta}e^{(-\beta\mu+\frac{\beta^2\tau^2}{2})}\Phi^c\left( \frac{\log{x}-\mu+\beta \tau^2 }{\tau}\right)
\end{equation}
\end{small}
with $E[X]=\mu- \frac{1}{\beta}$ and $VAR[X]=\tau^2 +
\frac{1}{\beta^2}$. In Section~\ref{sec:implication}, we will analyze the
CDF and formally prove that the Pareto component
describes the low values of the distribution and the Lognormal one dominates
the high values.

\para{Pareto-Lognormal Likelihood Derivation}
\begin{figure}[t]
\centering
\epsfig{figure=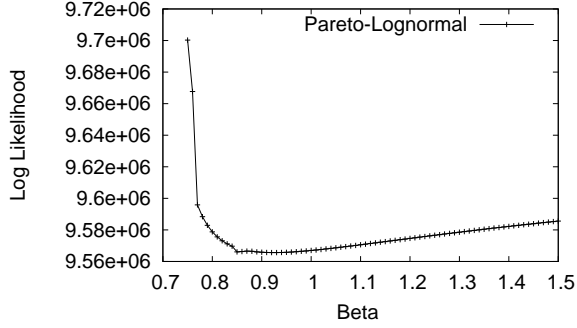,width=3.2in}
 \caption{\small The $\beta$
 value corresponding to the minimum reverse log likelihood value maximizes the 
 likelihood of the Pareto-Lognormal distribution to fit the London sample.}
\label{fig:likelihood}
\end{figure}
The likelihood function estimates the parameters of a distribution function in light of the observed data.
The likelihood $L$ is a function of parameters $\theta$ of the distribution and
is defined as follows:
$$L(\theta)=\prod_{i=1}^n {f(x_i | \theta)}$$
Using the definition of the Pareto Lognormal distribution in
Equation~\ref{pdf_pareto_log}, the likelihood becomes:
$$L(\beta,\mu,\tau)=\prod_{i=1}^n {\beta x_i^{\beta-1}e^{(-\beta\mu+\frac{\beta^2\tau^2}{2})}\Phi^c\left( \frac{\log{x_i}-\mu+\beta \tau^2 }{\tau}\right)}$$
The likelihood is defined as a product, and maximizing a product is usually
more difficult than maximizing a sum. Instead, we use a monotonously
increasing conversion function to transform the function $L(\beta,\mu,\tau)$
into a new function $L'(\beta,\mu,\tau)$, such that $L(\beta,\mu,\tau)$ and
$L'(\beta,\mu,\tau)$ have their maximum values for the same $\beta$, $\mu$ and
$\tau$ values.  The monotonous transformation we use is the logarithmic
function, which turns the maximization of a product into an easier
maximization of a sum.  For simplicity, let $A_0$ be
$-\beta\mu+\frac{\beta^2\tau^2}{2}$, then
$$\log{L}= \log{\prod_{i=1}^n {f(x_i)}}=\sum_{i=1}^n \log{f(x_i)}=$$
\begin{small}
$$\sum_{i=1}^n \log{\beta} +\sum_{i=1}^n \log{x_i^{\beta-1}} +\sum_{i=1}^n \log {e^{(A_0)}} 
\sum_{i=1}^n \Phi^c\left( \frac{\log{x_i}-\mu+\beta \tau^2 }{\tau}\right)=$$
\begin{equation}\label{log_likelihood}
n\log{\beta} +(\beta-1)\sum_{i=1}^n \log{x_i} + nA_0+\sum_{i=1}^n \Phi^c\left( \frac{\log{x_i}-\mu+\beta \tau^2 }{\tau}\right)
\end{equation}
\end{small}

\begin{figure*}[t]
\begin{minipage}[h]{0.66\columnwidth}
\centering
\epsfig{figure=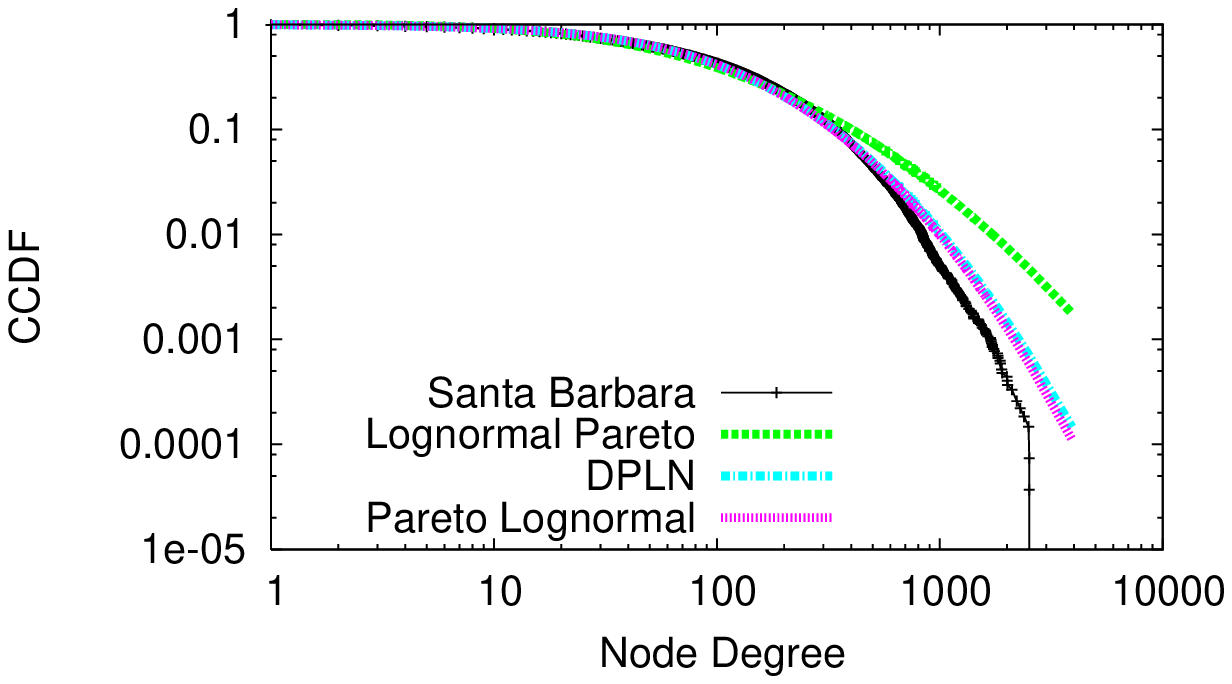,height=1.25in}
\end{minipage}
\hfill
\begin{minipage}[h]{0.66\columnwidth}
\centering
\epsfig{figure=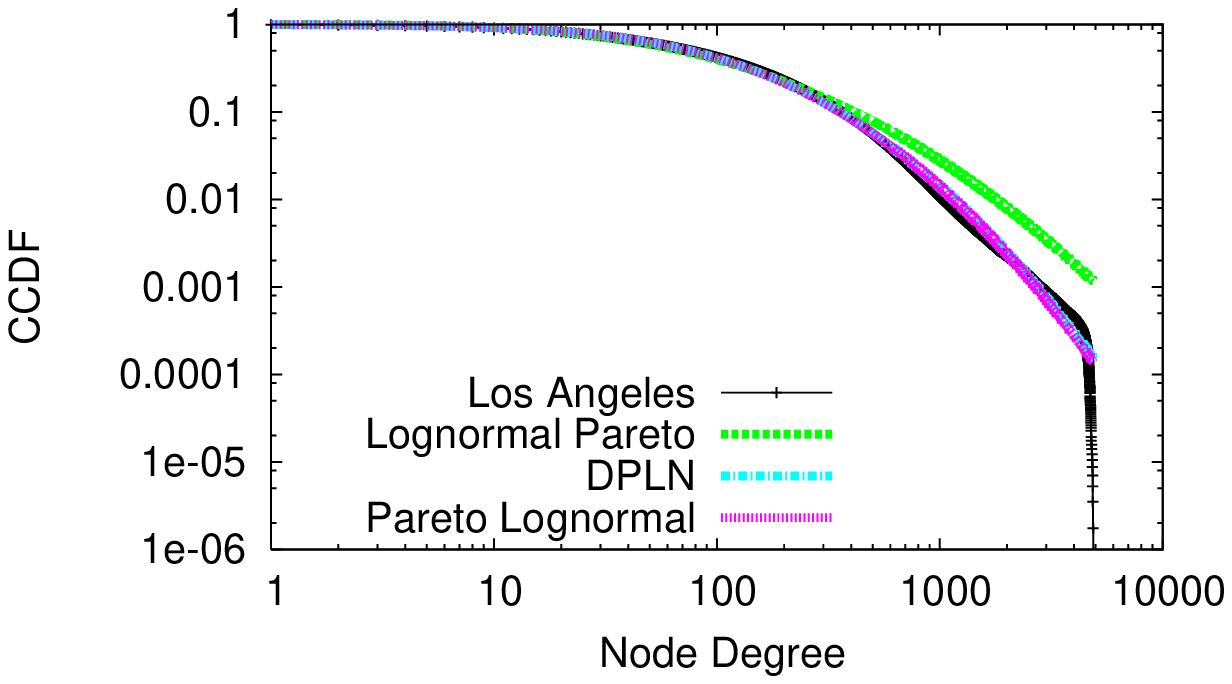,height=1.25in}
\end{minipage}
\hfill
\begin{minipage}[h]{0.66\columnwidth}
\centering
\epsfig{figure=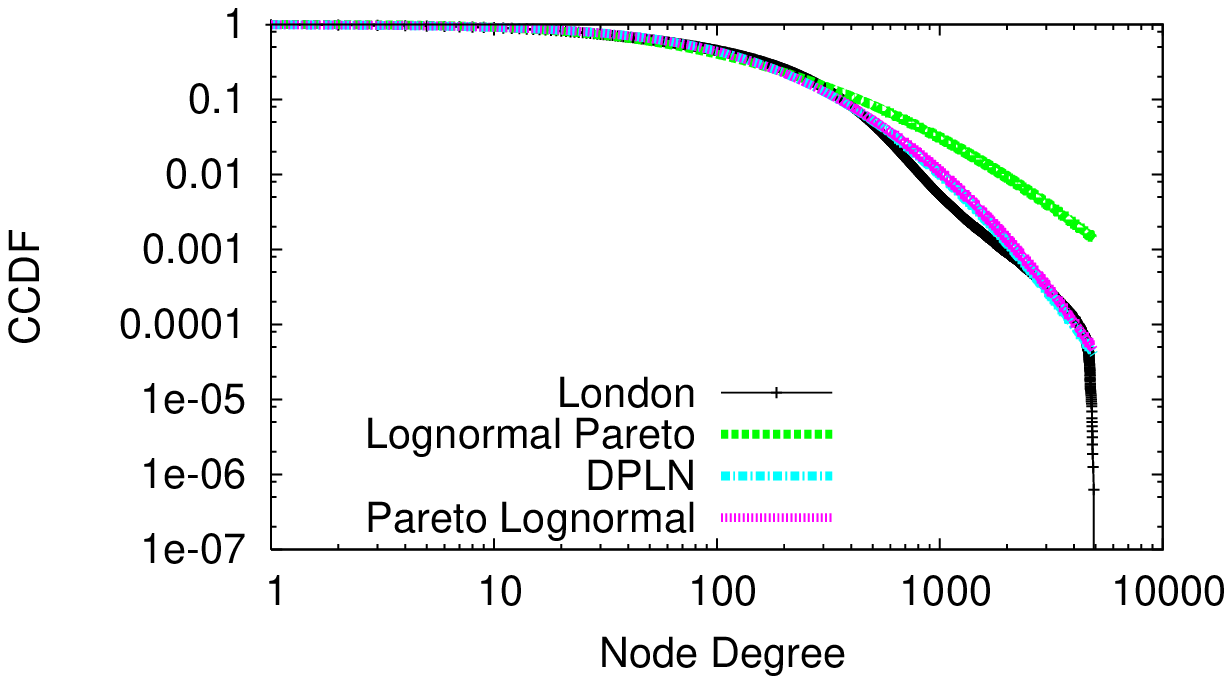,height=1.25in}
\end{minipage}
\caption{\small The CCDF of the DPLN and its two limit forms, {\em i.e.}
  Pareto-Lognormal (PLN) and Lognormal Pareto (LNP) fitted on Santa Barbara,
  Los Angeles and London networks. The PLN provides the same accurate fit as
  the DPLN but with much lower fitting complexity.}
 \label{fig:pdf_ccdf1}
\end{figure*}

We use~(\ref{log_likelihood}) to estimate the parameters $\beta$,
$\mu$ and $\tau$ in order to fit the Pareto-Lognormal model to our real OSN graphs.  We also {\em reverse the sign} of
(\ref{log_likelihood}) such that the likelihood to fit the data is
maximized by minimizing $-L'(\beta,\mu,\tau)$. 

\begin{figure}
\centering
\epsfig{figure=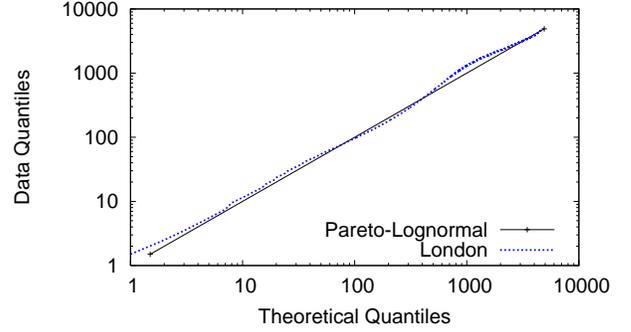,width=3.2in}
\caption{\small The Pareto-Lognormal Quantile-Quantile plot on London shows that
  this model almost perfectly approximates the real data.}
\label{fig:overhead-gia}
\end{figure}

\begin{table*}[t]
\begin{center}
\begin{small}
\begin{tabular}{|c|c|c|c|c|c|c|c|c|}
\hline
 Real  & Statistical &\multicolumn{7}{|c|}{Distribution Models} \\ \cline{3-9}
 Graph  &  Methods & Power-law & Lognormal & Exponential &  Pareto-Exp. & LNP & PLN & DPLN\\ \hline
 \hline
Monterey & Log L. & 8.75338e+04& 7.83008e+04 & 7.78214e+04 & 7.76411e+04 &  7.76411e+04  & {\bf 7.74057e+04} & 7.74089e+04\\
nodes $13.843$ & AIC & 1.75069e+05 & 1.56605e+05 & 1.55644e+05 & 1.55188e+05& 1.55288e+05  & {\bf 1.54817e+05} & 1.54825e+05 \\
edges $\approx 704$K & RSS &  3.05267e+02 & 2.07443e+00& 4.86935e-01& 3.43766e-01& 2.59320e-01  &{\bf 2.26463e-02} & 2.32905e-02 \\ \hline
\hline

Santa Barbara & Log L. & 1.80247e+05& 1.63328e+05 & 1.61209e+05 & 1.60516e+05 & 1.61152e+05  & {\bf1.60448e+05} & 1.60455e+05\\
nodes $27.140$ & AIC & 3.60497e+05 & 3.26660e+05 & 3.22421e+05 &3.21036e+05 & 3.22311e+05  & {\bf3.20902e+05} & 3.20919e+05 \\
edges $\approx 2$M & RSS &  4.38954e+02 & 5.81415e+00 & 7.31640e-01 & 3.37774e-01 & 6.10471e-01  & {\bf 7.44256e-02} & 7.74719e-02 \\ \hline
\hline

Egypt & Log L. & 1.62073e+06 & 1.53376e+06 & 1.53832e+06 & 1.50289e+06& 1.49892e+06& 1.49855e+06 & {\bf1.49830e+06} \\
nodes $283$K & AIC &3.24148e+06 & 3.06754e+06 & 3.07664e+06 & 3.00579e+06 & 2.99785e+06 & 2.99710e+06 & {\bf 2.99661e+06} \\
edges $\approx 11$M & RSS & 1.17928e+03 & 7.68335e+00& 1.69934e+00 & 3.43834e-01& 2.13671e-01 & {\bf 1.58090e-01} & 2.03417e-01\\ \hline
\hline

Los Angeles & Log L. & 3.81640e+06 & 3.45975e+06 & 3.44290e+06 & 3.42677e+06& 3.41280e+06 & {\bf 3.40372e+06} & 3.40379e+06\\
nodes $572$K & AIC & 7.63281e+06 & 6.91950e+06 & 6.88581e+06 & 6.85354e+06& 6.82561e+06 & {\bf 6.80745e+06} & 6.80759e+06  \\
edges $\approx 43$M & RSS & 1.74060e+03 & 8.60777e+00 & 1.50800e+00 & 8.55362e-01 & 5.02215e-01 & 4.82168e-02 & {\bf4.75582e-02} \\ \hline
\hline

New York & Log L. & 5.66194e+06 & 5.23512e+06 & 5.17802e+06 & 5.12808e+06 & 5.13177e+06 & {\bf 5.11728e+06} &5.11753e+06 \\
nodes $855$K & AIC & 1.13239e+07 & 1.04702e+07 & 1.03560e+07 & 1.02561e+07 & 1.02635e+07 & {\bf 1.02345e+07} & 1.02350e+07  \\
edges $\approx 66$M & RSS & 1.81230e+03 & 1.65956e+01 & 1.87038e+00 & 5.88628e-01 & 7.90849e-01 & 1.70891e-01 & {\bf 1.65433e-01} \\ \hline
\hline

Manhattan R.W.& Log L. & 6.76810e+06 & 5.92193e+06 & 5.86365e+06 & 6.02838e+06 & 5.89459e+06  & {\bf 5.85632e+06} & 5.85679e+06  \\
nodes $957$K & AIC &  1.35362e+07 & 1.18438e+07 & 1.17273e+07 & 1.20567e+07 & 1.92208e+07  & {\bf 1.17126e+07}  &  1.17135e+07 \\
edges $\approx  80$M & RSS & 1.31910e+03 & 4.01862e+00& 9.62000e-02 & 8.07448e+00 & 6.29269e-01  & 2.44948e-02 &  {\bf  2.34597e-02} \\ 
\hline

London & Log L. & 1.07131e+07& 9.75799e+06 & 9.60413e+06 & 9.65522e+06 & 9.61043e+06 & {\bf 9.56562e+06} & 9.56626e+06 \\
nodes $1.6$M & AIC & 2.14262e+07& 1.95159e+07 & 1.92082e+07 & 1.92104e+07 & 1.17892e+07  & {\bf 1.91312e+07}  & 1.91325e+07 \\
edges $\approx 118$M & RSS & 1.79927e+03 & 1.24657e+01& 7.33091e-01 & 2.30840e-01 & 1.26502e+00 & {\bf 1.99847e-01} & 2.12646e-01 \\ 
\hline

Orkut & Log L. & 1.81346e+07 & 1.63179e+07 & 1.62442e+07 &1.61573e+07 & 1.61822e+07  & 1.60868e+07 & {\bf 1.60829e+07} \\
nodes $3$ M & AIC & 3.62692e+07  & 3.26358e+07 & 3.24885e+07 & 3.23147e+07& 3.23645e+07  & 3.21738e+07 & {\bf 3.21659e+07} \\
edges $\approx 111$ M & RSS & 1.94280e+03 & 2.41064e+00 &3.75274e-01 & 2.95699e-01&4.16978e-01  & 1.60096e-02 & {\bf1.01767e-02} \\ \hline
\end{tabular}
\end{small}
\end{center}
\caption{\small Quantifying the ``Goodness of the fit'' of $7$ distribution
  models via statistical methods on $6$ Facebook datasets crawled in $2009$,
  and on the Orkut dataset.}
\label{table:stats}
\end{table*}

\para{Parameter estimation.}
The Pareto-Lognormal distribution has three parameters ({\em i.e.} $\beta$,
$\mu$ and $\tau$) to estimate in order to fit real data.  We perform a grid
search, {\em i.e.} a multi-dimensional numerical search over the parameters
space, to identify the best triplet of values to fit a particular dataset, as in ~\cite{mobile_dpln}.

We bound the search of the parameters $\mu$ and $\tau$ using the {\em
  Moment Method Estimation} to determine the initial values, and then 
refine the search around those values to identify the best ones. 
While this methodology could lead to suboptimal results that lower 
the performance of our models when compared to those using 
optimal parameter configurations, we choose it because of its computational efficiency.

We use the likelihood metric as the objective function in our parameter search~\cite{Arnold_pareto}, {\em i.e.\/} our goal is to 
minimize the reverse log likelihood ``$-L(\beta,\mu,\tau)$'' when
searching for the optimal ($\beta$,
$\mu$, $\tau$)~\cite{Arnold_pareto}.  
To show that there is a clear concave
trend around the minima, we plot in Figure~\ref{fig:likelihood} the values of $-L(\beta,\mu,\tau)$  for the London dataset as a function of $\beta$.
We have identified similar trends in all our datasets.

\para{Fitting to real datasets.} Figure~\ref{fig:pdf_ccdf1} examines the results of
fitting the PLN, DPLN and LNP distributions to our graphs as CCDF
plots. Among these three models, LNP overestimates the data on the
upper tail, which suggests that the decay of the right tail follows more of a
Lognormal model than a Pareto one.  The flexibility of the DLPN model, with
its four parameters, allows for a very good fit. However, it comes at the cost
of a higher fitting complexity which grows exponentially in the number of parameters. 
On the other hand, the PLN distribution
achieves the accuracy of the DPLN, and also has the reduced
fitting complexity of LNP ({\em i.e.} 3 parameters instead of 4).
Figure~\ref{fig:pdf_ccdf1} clearly shows that on all three datasets ({\em
  i.e.} Santa Barbara, Los Angeles and London), PLN and DPLN overlap along
the entire distribution.  In addition, we use the Quantile-Quantile plot 
to evaluate the fitting accuracy of PLN, and show the results for London 
in Figure~\ref{fig:overhead-gia}. This plot is representative and others 
are omitted for brevity.  It shows a near perfect fitting from the 
Pareto-Lognormal to the real London graph.

\subsection{Statistical Analysis}
\label{sec:statistical_analysis}
The graphical assessments in Section~\ref{sec:elementary}
 ({\em i.e.} the PDF, CCDF and the
Quantile-Quantile plot) are the first step towards a complete
characterization of these distributions.  We now look at statistical measures
to quantify how well each model fits real social graphs~\cite{model_selection}.
 
\para{Goodness-of-Fit Analysis.}
 We consider three
measures to evaluate the goodness of the proposed statistical models, including the likelihood function,
the Akaike's Information Criterion and the Residual Sum of Squares.

The {\em likelihood function} is used to quantify the likelihood
that a particular model fits a given dataset.
Section~\ref{sec:likelihood_formalization} explained how to 
estimate this function for PLN by computing the minimum of
the reverse log likelihood. We compute the minimum of the reverse of the log
likelihood function for each model such that the best fit models generate the
lowest values.  This matches the two other measures, where the best models also
generate the smallest fitting errors.
 
{\em Akaike's Information Criterion} 
(AIC)~\cite{Akaike} is a measure of the quality of
the fit that is capable of capturing the tradeoff between accuracy and fitting costs.
The AIC value is computed as: $AIC = 2k - 2\log{L}$ where $k$ is the number
of parameters in the statistical model, and $L$ is the maximized value of the
likelihood function for the estimated model.  The value $k$ in the AIC test
is used to tradeoff the accuracy of a model with its complexity. The model
that shows the lowest AIC is considered the one that best fits the data.

{\em Residual Sum of Squares}(RSS)~\cite{mobile_dpln} is a statistical method that computes the sum
of squares of residuals between the empirical distribution and the data sample.
It measures the discrepancy as euclidean distance between the data and the
estimation model. A small RSS
indicates a tight fit of the model to the data, and it can be formally expressed as: 
$RSS=\sum{(y_i-f(x_i))^2}$, where $y_i$ is the empirical evaluation and
$f(x_i)$ is the estimate value of the statistical model.

\para{Error Measures  on Real Traces.}\label{sec:fitting data} 
We now compare models based on the numerical values of these three 
statistical methods computed on our OSN datasets.  
We explore sample size variation to 
prove that the datasets manifest the same trend across different sample 
sizes.

In Table~\ref{table:stats} we report a set of statistical values to quantify
the goodness-of-the-fit for each of the {\em seven} analyzed models.  We
highlight in bold the smallest values ({\em i.e.} those that identify
the best model) for each metric.
  
Across all datasets, the Power-law model consistently performs the
worst.  Its values of Log Likelihood and AIC test are the highest and the RSS
values are up to $4$ orders of magnitude higher than the best model.  The second
worst model on our measured datasets is the Lognormal.  Its RSS values are up
to $3$ order of magnitude larger than the best model, due to the high
imprecision in estimating the high degree nodes. 
Exponential, Pareto with
Exponential cutoff and LNP provide reasonable accuracy
in the sense their values are with $1$ order of magnitude from the best.
 
The two best models are DPLN and PLN.  The RSSs
for some datasets identify DPLN as the best model, with a very small
difference separating it from PLN (on the second or third decimal point). 
On the other side, the AIC test slightly
penalizes the model with more parameters, such as DPLN. 
Based on the results of both the likelihood and
the AIC,  we see that PLN does consistently well on all our Facebook datasets. 
Finally, analysis of the Orkut graph produces results consistent to our Facebook
graphs. 

Overall we see a consistent trend: Power-law does not produce accurate
results, and the best models are PLN and DPLN. Only in a small number of
cases, DPLN is slightly more accurate than PLN, but differences are exceptionally small compared to other
models.  Given the significant increase in fitting
complexity for DPLN ({\em i.e.\/} deriving 4 rather than 3 parameters), the PLN model clearly produces the best combination of
fitting accuracy and complexity.

\section{Implications of the PLN model}
\label{sec:implication}

We have shown that the PLN model
provides a much more accurate fit to today's OSN graphs. But are these
differences large enough to really matter?
In this section, we answer this question by quantifying the magnitude of
error introduced by the Power-law model, and in the process, show that
social algorithms and protocols based on the Power-law assumption must be
re-evaluated (and some re-designed) using the PLN model.  

The remainder of this section includes three parts. First, we
analyze PLN using its CDF to characterize the asymptotical slope of
its tail, which we use later to derive more complete bounds,
in order to better understand the high degree nodes of these networks.
Second, for both  PLN and Power-law model, we analytically determine a
 close form to bound the lowest degree of high degree nodes. 
By comparing these bounds, we formally quantify the divergence of Power-law from PLN. 
We also validate our analytical results from empirical analysis on our
Facebook and Orkut graphs.
Finally, for both models, 
we approximate the cardinality of
high degree nodes (those with degrees in the top 10\% of the network), and evaluate the discrepancy between the two. 
Again, we validate our analytical results using our real datasets to
understand the actual prediction errors from both models.

\subsection{Modeling High Degree Nodes with PLN}\label{F_limit}
The cardinality of high degree nodes is a key factor in designing social applications and protocols. 
To characterize this distribution in the PLN model, it is necessary to understand the limit behavior on its CDF (or $F(x)$), 
defined by (\ref{cdf_pareto_log}).  We aim to bound the limit behavior so that we can directly
derive the expected number of high degree nodes and their connectivity. 
Let $z=\frac{\log{x}-\mu}{\tau} $,  $A =
e^{(-\beta\mu+\frac{\beta^2\tau^2}{2})}$, we have
\vspace{-0.1in}

\begin{small}
\begin{eqnarray}\label{F_of_x}  
F(x)=\Phi(z)+x^\beta A\; \Phi^c(z+\beta\tau)
\end{eqnarray}
\end{small}
Leveraging the {\em erf}
function ({\em i.e.} the Gauss error function, which is a special function of
sigmoid shape), we can express the standard normal cumulative distribution
$\Phi(x)$ and its complementary form $\Phi^c(x)$ respectively as follows:
\vspace{-0.15in}

\begin{small}
\begin{eqnarray}\label{approx_phi}  
\Phi(x)=\frac{1}{2} [ 1+erf(\frac{x}{\sqrt{2}}) ]=\frac{1}{2}erfc(-\frac{x}{\sqrt{2}})
\end{eqnarray}
\end{small}
\vspace{-0.1in}
\begin{small}
\begin{eqnarray}\label{approx_phic} 
\Phi^c(x)=1-\frac{1}{2} [ 1+erf(\frac{x}{\sqrt{2}}) ] =\frac{1}{2}erfc(\frac{x}{\sqrt{2}})
\end{eqnarray}
\end{small}

\vspace{-0.05in}
\noindent with $x \in \Re$. Note that by definition, $erf(-x)=-erf(x)$ and $erfc(x)=1-erf(x)$.
We can now reformulate (\ref{F_of_x})  
through the use of Gauss error functions as:
\begin{equation}\label{approx_F(x)}
F(x)=\frac{1}{2}erfc \left(- \frac{z}{\sqrt{2}}\right)+x^{\beta}A\frac{1}{2}erfc \left( \frac{z+\beta\tau}{\sqrt{2}}\right)
\end{equation}
Next we use the asymptotical expansion of the Gauss error function to further expand $F(x)$. For large $x$, the
$erfc()$ function can be approximated with the following series:
\vspace{-0.05in}
\begin{equation}\label{approx_erfc}
erfc(x)\approx \frac{e^{-x^2}}{x\sqrt{\pi}}\sum_{n=0}^{+\infty} (-1)^n\frac{2n!}{n!(2x)^{2n}}
\end{equation}
Without loss of generality, in the remainder of this analysis, we will only
use the first term of the series in (\ref{approx_erfc}). This is
because it achieves enough accuracy to accomplish the goals of this
investigation. While we omit a formal error analysis of the loss in the
approximation, we will show via experimental
analysis that the loss is negligible, and the results are sufficiently
accurate.  Thus the $erfc()$ function is approximated as follows:
\begin{equation}\label{approx_erfc_simple}
erfc(x)\approx \frac{e^{-x^2}}{x\sqrt{\pi}}
\end{equation}

\begin{lemma}\cite{Reed03thedouble}\label{lem:upper}
  For a sufficiently large $x$, $1-(\Phi(z)+x^\beta A \Phi^c(z+\beta\tau))$
  goes to $\Phi^c(z)$.
\end{lemma}
The claim states that for large $x$ the PLN distribution has a Lognormal form with the parameters of the
PLN model. We will use this Lemma extensively in the following discussions.

\subsection{Quantile Analysis: a Degree Threshold}\label{quantile}

Next, we quantify how well the PLN and the Power-law models
predict the degree threshold that defines the top $\gamma$ subset of the highest
degree nodes in the network. In other words, we wish to predict the minimum
degree $\xi_{\gamma}$, such that nodes with degree $> \xi_{\gamma}$ are in
the top $\gamma $ portion of all nodes sorted by degree.  Formally, $\xi_{\gamma}$ is
the $\gamma$-th quantile of the complementary cumulative distribution. As a
concrete example, we will use $\gamma=0.10$, {\em i.e.} top 10\%.

We compute $\xi_{\gamma}$ for both models.  For the Power-law model, the
degree $\xi_{\gamma}$ can be expressed as:
$(\frac{1}{\gamma})^{\frac{1}{\alpha}}$.  To quantify the same
quantile for the PLN distribution, we again leverage its limit behavior. We then use the asymptotic expansion of the
complementary Gauss error function to obtain a tight upper bound 
in Lemma~\ref{lem:xi}.
 
\begin{lemma}\label{lem:xi}
A tight upper bound of the $\gamma$-th quantile of the PLN distribution is  
$\xi_{\gamma}=e^{\mu+\sqrt{-2\tau^2 \log{(\frac{\sqrt{2 \pi}\gamma}{\tau})} }}$.
\end{lemma} 
\begin{proof}
  Given the result in Lemma~\ref{lem:upper}, we consider $1-F(x) \approx
  \Phi^c(\frac{\log{x}-\mu}{\tau})$ for large values of $x$. Since we are
  using the Pareto-Lognormal parameters, we can reformulate the
  complementary cumulative distribution function for large $x=\xi_{\gamma}$
  values as: \\
  $\frac{1}{2\sqrt{\pi}}\frac{e^{-(\log{\xi_{\gamma}}-\mu)^2/(2\tau^2)}}{\frac{\log{\xi_{\gamma}}-\mu}{\sqrt{2\tau^2}}}$.
  In order to derive the minimum degree $\xi_{\gamma}$, which characterizes
  the high degree nodes, we first quantify the $\gamma$\% of high degree
  nodes in the network, and then derive the required degree value.  We
  leverage the tail of the complementary cumulative distribution from
  Lemma~\ref{lem:upper}, and thus, the following holds:
\begin{equation}\label{eq:ccc}
\frac{1}{2\sqrt{\pi}}\frac{e^{-(\log{\xi_{\gamma}}-\mu)^2/(2\tau^2)}}{\frac{\log{\xi_{\gamma}}-\mu}{\sqrt{2\tau^2}}}=\gamma
\end{equation}
Let $y=\log{\xi_{\gamma}}-\mu$, then (\ref{eq:ccc}) becomes
$\frac{\tau}{\sqrt{2 \pi}} \frac{e^{-y^2/(2\tau^2)}}{y}=\gamma$.
Applying the logarithmic function and approximation with the quadratic term,
it reduces to $y=\sqrt{-2\tau^2 \log{(\frac{\gamma
      \sqrt{2\pi}}{\tau})}}$. By substituting $y$, we have
$\log{\xi_{\gamma}}-\mu=\sqrt{-2\tau^2 \log{(\frac{\gamma
      \sqrt{2\pi}}{\tau})}}$ or $\xi_{\gamma}=e^{\mu+\sqrt{-2\tau^2 \log{(\frac{\sqrt{2
          \pi}\gamma}{\tau})} }}$ which is the minimum degree of the high degree nodes in the $\gamma$-th quantile.
\end{proof}


Next, we approximate the differences among the $\xi_{\gamma}$ quantiles 
between the PLN and the Power-law distributions. 
Note that we are not approximating through a limit formulation but we are estimating 
the difference of the two, i.e. PLN and Power-law, analytical quantile values when 
$\gamma$  $\in [0.01,0.1]$. 
In order to compute the difference of the $\xi_{\gamma}$ estimated by the
Power-law and the PLN, we approximate the PLN's $\xi_{\gamma}$ as $e^{\mu}$,
since $\sqrt{-2\tau^2 \log{(\frac{\sqrt{2 \pi}\gamma}{\tau})}}$ is negligible
compared to $\mu$ when $\gamma$ is around its typical value $0.05$.

The Power-law overestimation can be computed as the ratio of the $\xi_{\gamma}$ quantiles
of the Power-law and the PLN.  
Since $\mu$ and $\frac{1}{\alpha}$ are approximated by
the mean value, $\nu$, of the sample logarithms,  the following theorem
applies:
\begin{theorem}\label{theo:quantile}
For $\gamma \in [0.01,0.1]$, the ratio of the $\xi_{\gamma}$ quantiles of the Power-law and PLN
is  $\approx(\frac{1}{e \gamma})^{\nu}$.
\end{theorem}
%
%
Both the PLN and Power-law estimates will express the $\xi_{\gamma}$ degree
as an exponential function, but they have different bases. 
For the
Power-law, the base of the exponent is $\frac{1}{\gamma}=10$; for PLN it is
$e$.  This discrepancy accounts for the large difference in the predictions
made by the two models. As we will show using experimental validation on
real datasets, this difference can be as large as two order of magnitude.

\begin{table}[t]
\begin{center}
\begin{small}
\begin{tabular}{c|c|c|c}
\hline
Real  &\multicolumn{3}{|c}{  $\xi_{\gamma}$}(see Lemma $2$) \\ \cline{2-4}
Graphs & Power-law & PLN & Real Data \\ \hline
 
Monterey & 8906 & 582 &246\\ 
Santa Barbara & 160362 & 845 & 345\\ 
Egypt  & 29987 & 775 & 206\\
Los Angeles  & 170457 & 1101 & 363\\ 
New York  & 152937  & 1099 & 395\\
Manhattan R.W. & 357804 & 797& 583 \\
London & 176742 & 939 & 364\\ 
Orkut &41377&369&155\\
\end{tabular}
\caption{\small Comparing the min degree of
  the top $10\%$ high-degree nodes on our datasets, against the values
  predicted by the theoretical bounds.}
\label{table:stats3}
\end{small}
\end{center}
\end{table}

\para{Experimental Validation on Real Data.} Now we aim to
validate the theoretical result shown in the previous section using
experimental analysis on our datasets.  For each dataset
presented in Table~\ref{table:stats}, we compute its $\xi_{\gamma}$ value
({\em i.e.} the estimated minimum node degree within the $10$\% of the
highest degree nodes), and compare them against
their analytical bounds from the PLN (see Lemma~\ref{lem:xi}) and Power-law
models.

We list the results in Table~\ref{table:stats3}. Clearly, our theoretical approximation of
$\xi_{\gamma}$ values using PLN are consistently more accurate than those from the
Power-law model, which over-estimates the real data up to $3$ orders of
magnitude.  For instance, on the New York graph, the $\xi_{\gamma}$ value from
the real sample is $395$; our theoretical approximation from PLN is $1099$;
but the Power-law predicts $152937$.

\subsection{Cardinality of High-degree Nodes}\label{cardinality}
We next use the distribution models to predict the cardinality of
high-degree nodes in OSN graphs.  The cardinality of high-degree nodes
is commonly used to design social algorithms and protocols, and to 
evaluate their performance and complexity. 
In the following, we first analytically quantify this metric using both the Power-law and PLN
models, and then empirically validate our predictions using
real datasets.

\para{Power-law model.}
Using the Power-law density function, we can derive the number of high degree
nodes by computing the integral of the upper tail of the distribution. Let
$\xi$ be the minimum node degree among the high degree nodes, then the number
of high degree nodes is approximated as:

\begin{equation}\label{powerlaw_Snodi}
N\int_{\xi}^\infty c x^{-\alpha}dx =c N\frac{1}{\alpha-1}{\xi}^{-\alpha+1}\approx \frac{N}{{\xi}^{\alpha-1}}
\end{equation}
where $c$ is the normalization constant, approximated by 
$(\alpha -1)d_0^{\alpha -1}$, and $d_0$ is the minimum node degree in the network.

\para{PLN model.}
As mentioned before, a high degree node has a degree at least $\xi$, where
$\xi$ can be computed as in Section~\ref{quantile}. Thus, the next lemma
follows:

\begin{lemma}\label{lem:high_degree}
  Let $\xi$ be the minimum degree for a high degree node, then the number of high
  degree nodes $\approx \frac{N}{2 \pi}\frac{e^{-(\frac{\log{\xi}-\mu}{\sqrt{2\tau^2}})^2}}{\frac{\log{\xi}-\mu}{\sqrt{2\tau^2}}}$.
\end{lemma}

Using Lemma~\ref{lem:upper}, the number of high degree nodes described
by the PLN distribution can be studied from the $\Phi^c(x)$ component of the
CCDF of the PLN, where $x$ is $\frac{\log{\xi}-\mu}{\tau}$. Thus this number
becomes: $N\Phi^c \left(\frac{\log{\xi}-\mu}{\tau}\right)$. Using the
approximation in (\ref{approx_erfc}), we approximate the number of
nodes in the network with degree no less than $\xi$ by
$\frac{N}{2 \pi}\frac{e^{-(\frac{\log{\xi}-\mu}{\sqrt{2\tau^2}})^2}}{\frac{\log{\xi}-\mu}{\sqrt{2\tau^2}}}~.$

\begin{table}[t]
\begin{center}
\begin{small}
\begin{tabular}{c|c|c|c}
\hline
Real  &\multicolumn{3}{|c}{  \# of nodes with degree $\geq \xi_{\gamma}$} \\ \cline{2-4}
Graphs & Power-law & PLN & Real Data \\ \hline
Monterey & 1843 & 1387 &1385\\ 
Santa Barbara & 6763 & 2442 & 2724\\ 
Egypt & 61170 & 30387 & 28319 \\
Los Angeles  &141980& 54894  & 57182 \\ 
New York  &  204540 &77879  &85581\\
Manhattan R.W. & 76133& 76133& 72854 \\
London & 398529 & 156850 &160050 \\ 
Orkut &761970&334900&307240\\
\end{tabular}
\caption{\small Number of nodes with degree higher than $\xi_{\gamma}$, with $\xi_{\gamma}$ estimated from the $10\%$ of the nodes,  
computed on the fitted models and the real data.}
\label{table:stats4}
\end{small}
\end{center}
\end{table}

\begin{figure}[t]
\centering
\epsfig{figure=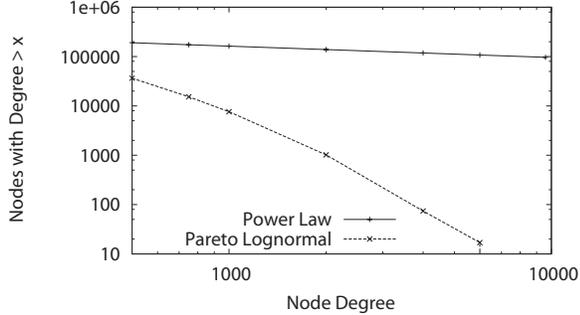,width=3in}
\caption{\small When predicting \# of high degree nodes, results from the Power-law and
  PLN models diverge up to $5$ orders of magnitude.}
\label{fig:supernodes}
\end{figure}

\begin{figure}[t]
\centering
\epsfig{figure=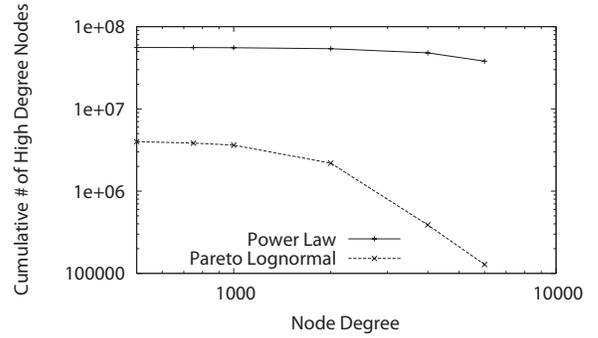, width=3in}
\caption{\small Comparing the cumulative degree of high degree nodes between the
  Pareto-Lognormal model and the Power-law reveals different values up to $3$
  orders of magnitude.}
\label{fig:degree_supernodes}
\end{figure}

\begin{figure*}[t]
\subfigure[Santa Barbara, CA]
{\epsfig{figure=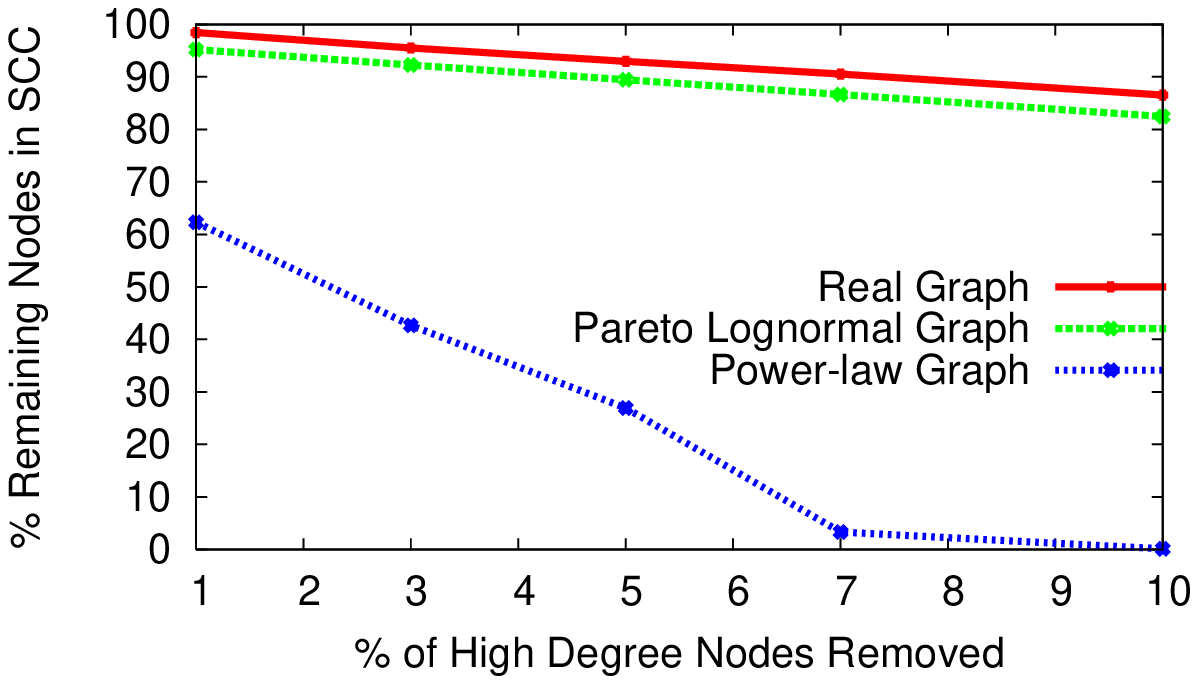, height=1.25in}}
\subfigure[Los Angeles, CA]
{\epsfig{figure=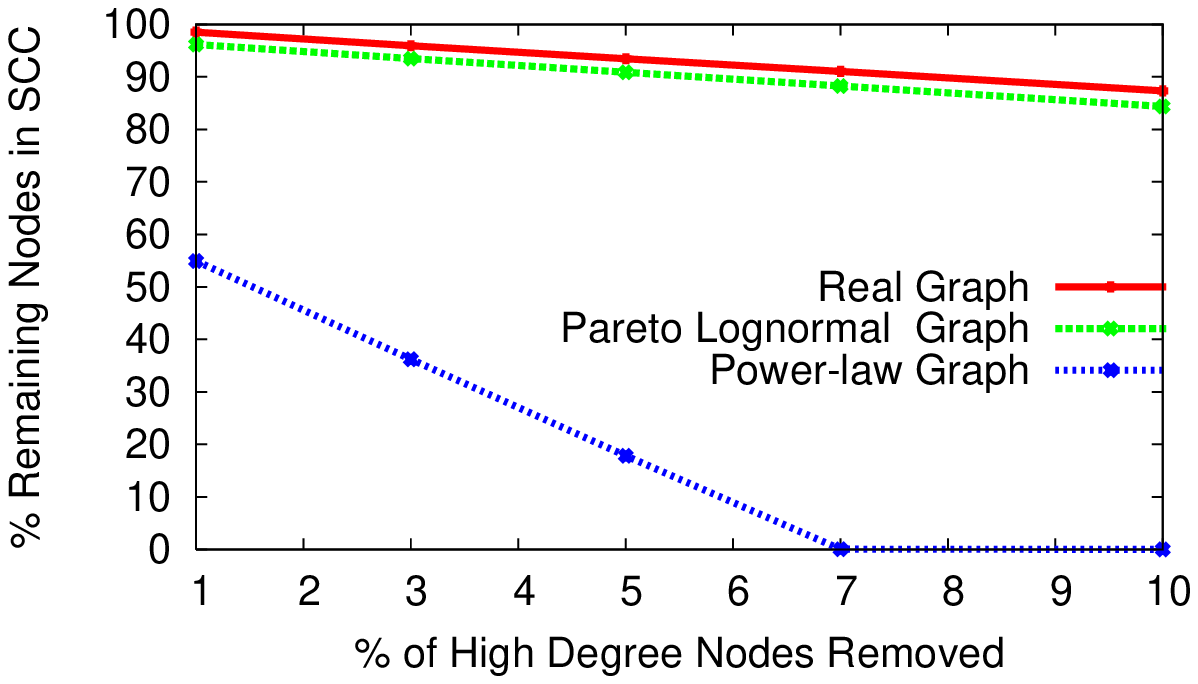,height=1.25in}}
\subfigure[London, UK]
{\epsfig{figure=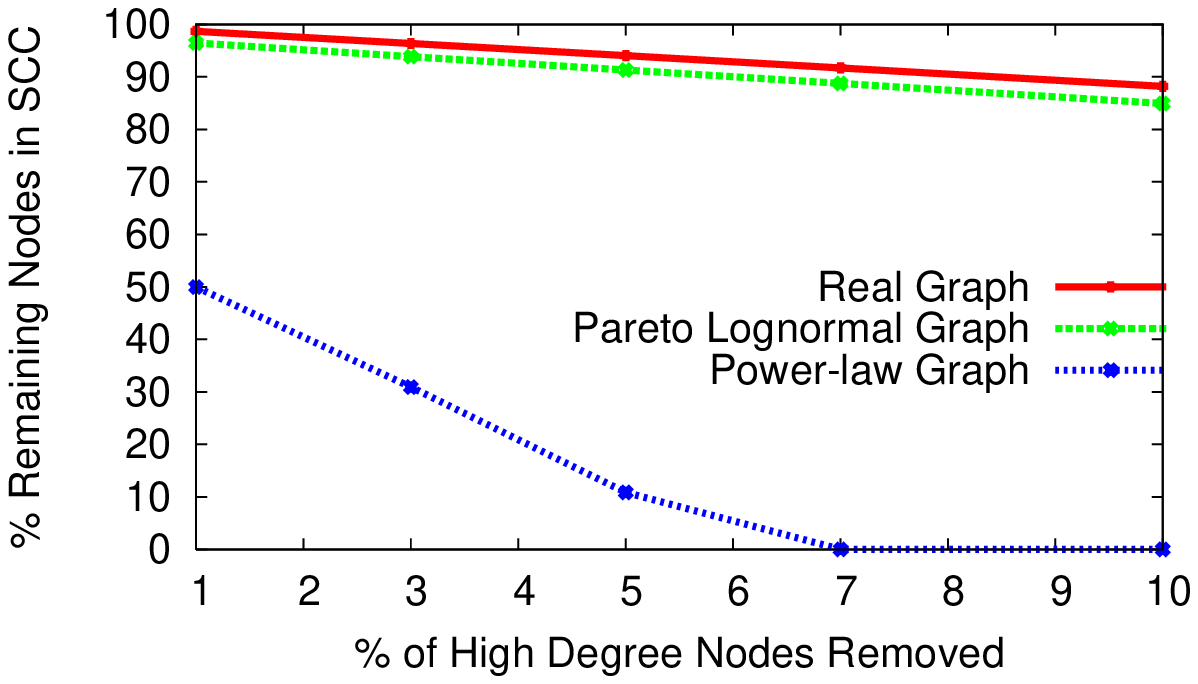,height=1.25in}}
\caption{\small Number of users remaining connected after deleting $x\%$ of
  the highest degree nodes from the graph.  The reported results compare
  partitioning effects on the synthetic PLN and Power-law graphs and a real
  graph.  The original graphs are Facebook graphs representing users in Santa
  Barbara CA, Los Angeles CA and London UK.}
\label{fig:cut_heuristic}
\end{figure*}

\para{Comparing High-degree Node Cardinality and Connectivity.}

We use two experiments to validate the quality of predictions on high-degree
nodes from PLN and Power-law.  First, we study the number of high-degree nodes
in the {\em measured} network datasets, and compare this number with that
estimated from the fitted models in Table~\ref{table:stats4}.
Specifically, we study the total number of high-degree nodes, defined as the nodes with
degree higher than $\xi_{\gamma}$ generated from the top $10\%$ of the
network as proven in Section~\ref{quantile}.

Table~\ref{table:stats4} shows the cardinality of high-degree nodes in the
measured networks, as well as  those predicted by the PLN and
Power-law models. Clearly, the PLN model provides a tight approximation of the real
dataset values, while  predictions made by the Power-law model
generally lead to more than $100$\% estimation error.

Empirical results like these can sometimes be biased because of specific
distributions in the real data.  To eliminate any possible bias, we generate
{\em pure sample} networks from the Pareto-Lognormal and the Power-law
models, then compute and compare each model's predicted number of high degree
nodes and their connectivity.

The {\em pure sample} networks are generated as follows.
For the Power-law, we take the parameters of a fitted model on a particular
real sample, and generate a pure sample by uniformly extracting $N$ numbers
between $0$ and $1$. Then for each extraction $x\in N$, a node with the
following degree is generated: $degree_x=\frac{1}{u^{(1/\alpha)}}$, where
$\alpha$ has the value of the Power-law exponent fitted on a real dataset.
For the Pareto-Lognormal, we generate a pure sample as follows: let $x$ be a
random number uniformly extracted between $1$ and $N$, and let $y$ be a
random number from the exponential distribution with mean parameter
$1/\beta$. Then $degree_{xy}=e^{\mu+\tau*x- y}$, where $\beta$ , $\mu$ and
$\tau$ are the parameters of the Pareto-Lognormal fitted on the real dataset.
In Figures~\ref{fig:supernodes} and~\ref{fig:degree_supernodes}, we compare
the number of high degree nodes and the sum of all node degrees respectively.
We do this for $N= 80K$, and use the average of the fitted model values as
the parameters of the Power-law and the Pareto Lognormal.
Figure~\ref{fig:supernodes} shows the divergence of these two models in
estimating the high degree nodes. In particular for nodes with degree
$>2000$, Power-law generates $2$ orders of magnitude more nodes than PLN.

\section{Impact on Social Applications}
\label{sec:application}

In the last section, we argued using both analytical predictions and
empirical data that the Pareto-Lognormal model provides a much more accurate
representation of node degree distributions in OSNs.  In this section, we seek
to better understand how the choice of degree distribution model impacts the
performance of applications on social graphs.  More precisely, for
applications that run on social graphs, we wish to quantify just how much
the choice of a degree distribution model alters their experimental
results.

When actual measured social graphs are not available, social applications
generally use synthetic graphs generated using statistics extracted from real
graphs~\cite{influence_kdd09}.  To highlight the differences when
applications use either the Pareto Lognormal or the Power-law degree model,
we take each of our real Facebook social graphs, and generate two synthetic
versions of them: one assuming a PLN model for degree
distribution, and one assuming Power-law.  We generate these
synthetic graphs using prior works from~\cite{Newman01arbitraryDegree, gen_degree_distr_graphs},
which can generate a synthetic graph with no self loops, given a specified
degree distribution and a given network size.

\begin{figure*}[t]
\subfigure[Independent Cascade]
{\epsfig{figure=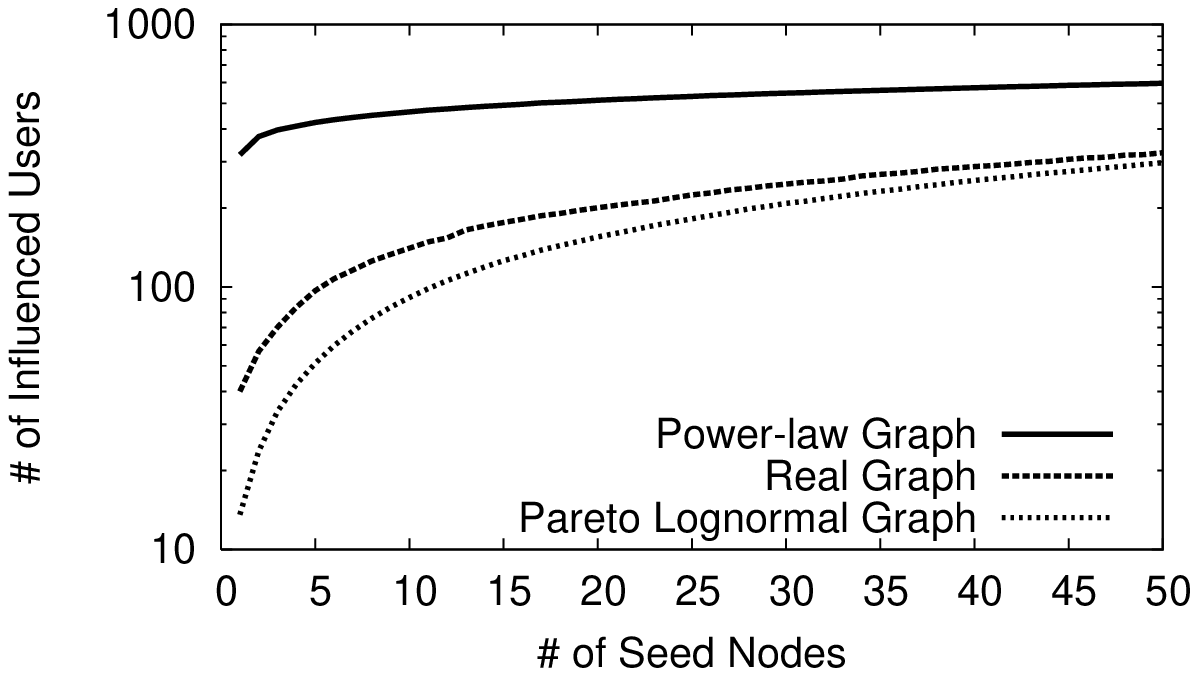, height=1.25in}}
\subfigure[Weighted Cascade]
{\epsfig{figure=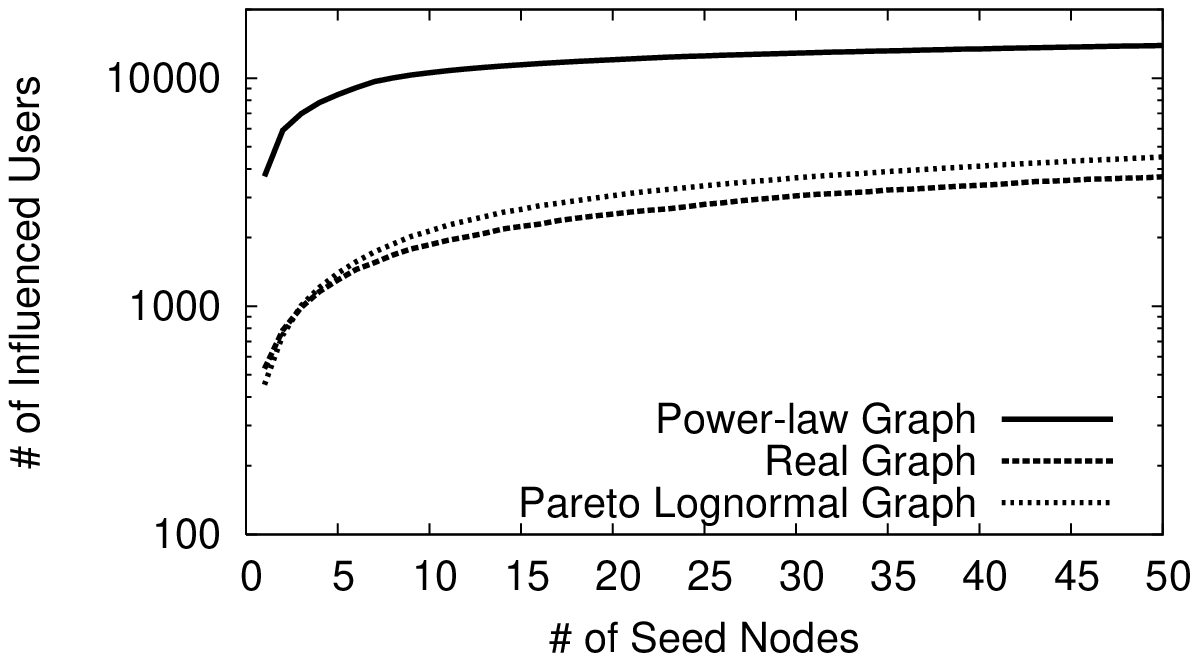,height=1.25in}}
\subfigure[Linear Threshold]
{\epsfig{figure=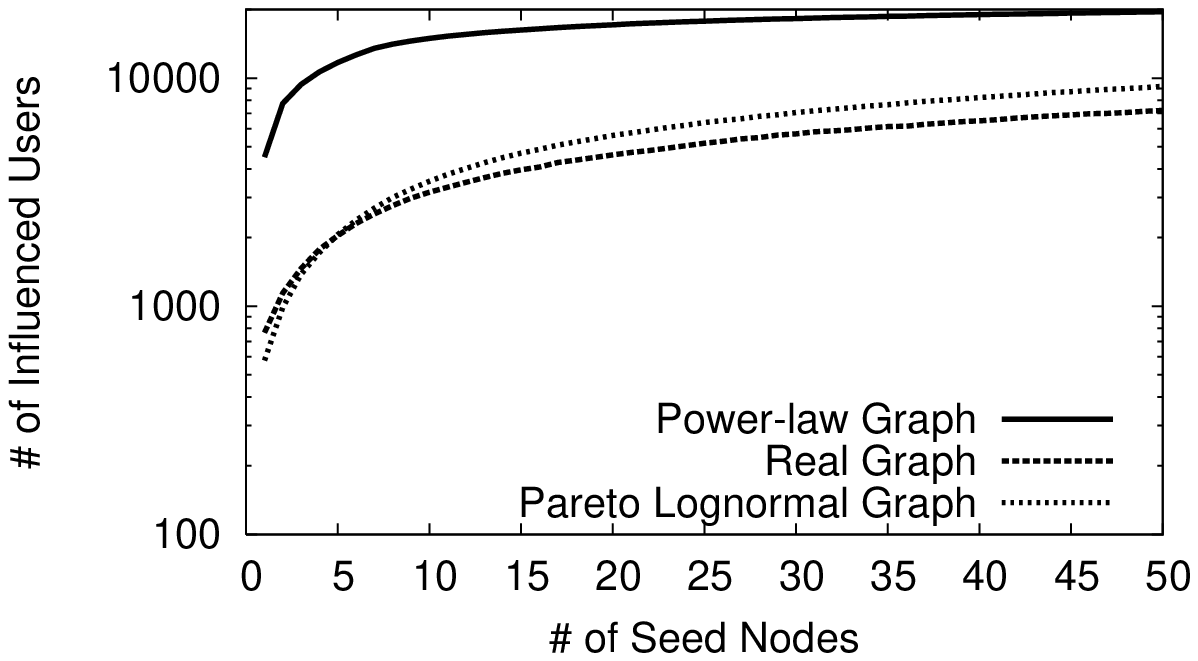,height=1.25in}}
\caption{\small Number of users influenced by a given number of seed nodes under
  three different influential dissemination models. Seed nodes are sorted by
  degree in decreasing order.}
\label{fig:maxinfluence}
\end{figure*}

We present results from implementations of three important applications on
social graphs.  They include: a) a graph partition approach that replicates
high degree nodes across partitions (introduced in Section~\ref{sec:intro}),
b) influence maximization on social graphs using three different
information spread models~\cite{influence_kdd09}, and c) a link-privacy attack
on anonymized social graphs~\cite{link_privacy}.  

All results are computed on a local cluster of Dell Xeon servers with 24--32GBs of
main memory.  Memory constraints limited the size of graphs we used in our
results.  Using Amazon EC2 large memory machines, we verified that our
observed trends hold for larger graphs such as the Orkut graph, but such
computations were too slow and costly in resources for us to generate
comprehensive graphs.

\subsection{Partitioning via Supernode Replication}
Efficiently answering queries on large graphs is a difficult challenge faced
by companies dealing with large network datasets, {\em e.g.} Facebook, Zynga.
Scalable solutions require splitting the graph data across machines in
computing clusters.  While some systems~\cite{little_engine,mapreduce}
distribute nodes randomly across a cluster, an ideal solution would find a
way to distribute the graph across the cluster as subgraphs with minimal
edges between them.  The best solutions would minimize edges between
partitions, which minimizes data dependencies between machines and maximizes
parallel query processing.

Unfortunately, graph partition is known to be
NP-Complete~\cite{graph_partition}, and social graphs are known to be very
dense graphs that do not partition well.  
Fortunately, Power-law networks are known to be vulnerable to ``targeted
attacks,'' {\em i.e.} they quickly fragment into disconnected subgraphs when
nodes with the highest degree are removed~\cite{powerlawfailure}.  Since
social graphs are widely accepted as Power-law graphs, we hypothesize we can
easily partition a social graph by first fragmenting it into subgraphs,
through the removal of a small number of supernodes.  These supernodes (and
their edges) can be replicated to every node in a cluster.  In addition to
potentially partitioning highly connected graphs, this approach is
attractive because it does not require the entire graph to
be in memory, and can thus efficiently run on extremely large social graphs.

We now look at the impact of running this partition scheme both on real
social graphs, and on synthetic graphs generated using the Power-law and PLN
synthetic models.  We take three of our social graphs, the Facebook regional
graphs from Santa Barbara, CA, USA and London, UK~\cite{interaction}, and the
Orkut social graph~\cite{socialnets-measurement}, and create synthetic graphs
that match each of them in size and node degrees, based on either the
Power-law or the PLN model for degree distributions.  We generate synthetic
graphs from degree distributions using Newman's
approach~\cite{Newman01arbitraryDegree}.

For each of these graphs, we test the effectiveness of our graph partitioning
strategy, by incrementally removing nodes from the graph, starting with those
nodes with the highest degree.  After these super nodes are removed, we
examine the remaining graph, and measure the size of the largest connected
component as a function of the original full graph.  

We plot the results in Figure~\ref{fig:cut_heuristic} for each of the
original graphs.  The results are strongly consistent among our datasets.
Regardless of the original graph in question, synthetic graphs produced using
a Power-law distribution are quickly fragmented. Removing just the top 3\% of
the highest degree nodes reduces the largest connected component to only
$40\%$ of the original graph.  In contrast, the real social graphs and
synthetic graphs from the Pareto-Lognormal distribution are highly resistant
to the fragmentation.  Even after removing $10\%$ of the highest degree
nodes, the graph remains largely connected, and the biggest connected
component still contains $80\%$ of all nodes. In this case, relying on an
inaccurate degree distribution model produces results fundamentally different
from the original graph.

\subsection{Influence Maximization}
Social networks have proven to be exceptionally useful tools for information
dissemination and marketing, and are used by companies and individuals to
promote their ideas, opinions and products.  One critical problem of interest
is that of influence maximization, or how to identify an initial set of users
who can influence the most number of users.  This is known as the {\em
  influence maximization problem.}

We examine the impact of degree distribution models on the influence
maximization problem.  Prior works have provided algorithms that use
statistical methods to model information dissemination over social
links~\cite{influence_kdd09, max_influence}.  We consider three different
models to spread information: {\em independent cascade model}, {\em weighted
  cascade model} and {\em linear threshold}.  These models differ in how they
compute the probability of influencing a node in the graph. For example, {\em independent
cascade} assumes a node is influenced independently from nearby nodes, while
the probability of a node in {\em weighted cascade} being influenced is a
function of its degree. In each case, heuristics provided by
\cite{influence_kdd09} allow us to compute the number of users influenced by
an initial set of ``seed'' users. 
We seek to understand how the spread of influence in these models is impacted
by the structure of the network.  We take the Santa Barbara Facebook graph,
its two synthetic variants based on Power-law and PLN, and compute how many
users are influenced as the number of seed nodes increases for each of the
three influence models mentioned above.  We focus on results from the Santa
Barbara network, since it is the largest graph we can execute given the
significant memory footprint of code from~\cite{influence_kdd09}.  Prior work
has shown this graph to be a representative social graph in all graph
metrics~\cite{modeling_www}.

Figure~\ref{fig:maxinfluence} shows the spread of information for each of the
three different models on the Santa Barbara network. For each model, we
compare the results on Santa Barbara with results on its Power-law
and PLN synthetic graphs.  
Our results are consistent across all three influential models. While the
exact result differs for each particular influential model, results from the
PLN graph results always closely follow results from the real graph. In
contrast, results on the Power-law graphs always overestimate the number of
influenced users on the real graph, sometimes overestimating by an order of
magnitude.  These results are further confirmation that using an inaccurate
degree distribution model can introduce dramatic errors in application-level
experimental results.

\begin{figure}[t]
\centering
\epsfig{figure=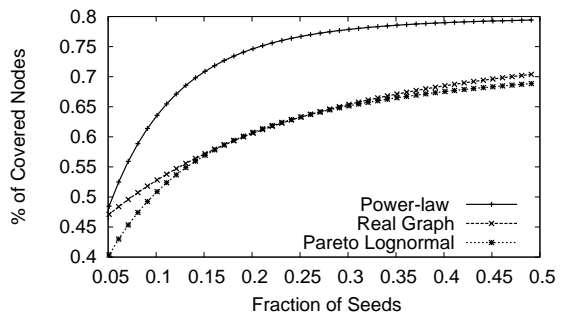, width=3in}
\caption{\small The capability to spread an attack through the network is
  highly dependent on assumptions of the underlying topology. The attack is
  much more effective on a Power-law network than both the real graph and a
  Pareto-Lognormal network.}
\label{fig:link_p}
\end{figure}

\subsection{Attacks on Link Privacy}
For our third application study, we focus on the problem of link privacy in
social graphs. Prior studies such as \cite{link_privacy} used experiments to
quantify the impact of attacks to disclose the presence of
connections between social network users.  This example differs from our
other application studies in that the application results are simulated based
on analytical derivations that integrate a degree distribution model.

\cite{link_privacy} presents different approaches to disclosure network
information by attacking particular nodes in the network.  The most effective
attack strategy is ``Highest,'' which tries to compromise nodes by spreading
the attack from high degree nodes.  Intuitively, because high degree nodes
have many incident edges, they are more likely to successfully spread the
attack across other nodes.  The authors provide a mathematical form to
quantify the effect of this attack by using a probabilistic approach that
leverages the probability distribution function of the network\footnote{ This
  use of network topology models is also commonly used in problems involving
  information propagation or dissemination in social networks.}.  This leads
to a simple question: how much does the effectiveness of this attack depend on
the choice of network topology model?

Quantifying the strength of this attack relies a notion of node
coverage.  A node is covered if and only if its $1$-hop neighbors are
unequivocally known. This definition allows \cite{link_privacy} to present
results where the strength of this attack is quantified as the fraction $f$
of nodes whose entire immediate neighborhood is known.

Let us begin by introducing the required variables to repeat the experiment:
$N$, $m$ and $d_0$ are respectively the total number of nodes in the
networks, the sum of the node degrees and the minimum degree.  The variable
$D=\sum_{i=1}^k d(x_i)$ identifies the sum of the degree of the $k$ nodes
from where the attack originates.  Theorem~$3$ in~\cite{link_privacy}
states that if $D= - \frac{\ln{\epsilon_0}}{d_0}2m$, then the disclosed nodes
after the attack are: $N(1-\epsilon-o(1))$.  $\epsilon$ is a variable
that we can compute as follows: let $\epsilon_0
=e^{(-\frac{d_0}{2m}\sum_{i=1}^k d(x_i) )}$, which is needed to estimate the
fraction of covered nodes.  We can derive $\epsilon$ as:
$$\epsilon = \sum_{x=d_0}^{k_{max}} e^{ -(\frac{x}{2m}\sum_{j=1}^k d(x_j)) } f(x)$$
where $k_{max}$ is the maximum node degree, $d(x)$ is the degree of $x$ and
$f(x)$ is the density function of the degree distribution. If we assume a
Power-law degree distribution, we substitute $f(x)$ with $cx^{-\alpha}$ and
in the case of the Pareto-Lognormal we use $f(x)$ as defined in (\ref{pdf_pareto_log}).

Figure~\ref{fig:link_p} shows that results on a Power-law network
significantly overpredicts the impact of attacks using the ``Highest''
strategy.  In contrast, results on the real social graph closely match those
from a synthetic graph following the Pareto-Lognormal degree model.  Clearly,
applications that use the Power-law model in their analytic derivations and
simulations are also significantly affected by their choice of node degree
distribution models.

\begin{figure}[t]
\centering
\epsfig{figure=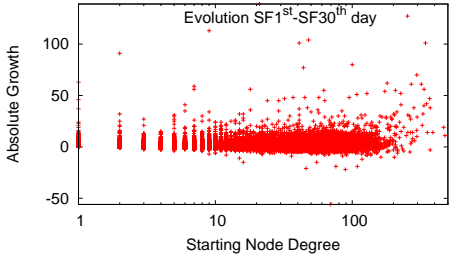, width=2.5in}
\caption{\small The evolution process affects nodes independently of their node
  degree.}
\label{fig:evol_ratio}
\end{figure}

\section{Towards a Generative Model}
\label{sec:generative}

In this section, we present preliminary results towards a generative model able to 
reproduce a Pareto Lognormal graph.
We also show using evidence in our datasets that the formation of these graphs 
diverges from the hypothesized {\em preferential attachment}-based scheme
presented in prior studies~\cite{Vazquez2003, Leskovec2007}.

Our goal is to provide an intuition of an algorithm that
captures a {\em lognormal multiplicative process}, and accurately models the
temporal evolution of online social networks.  Unlike generative models that
focus on reproducing a single snapshot of a network~\cite{Leskovec2007}, we
focus on modeling the evolutionary process that captures the network growth.
While the Power-law curve is generated by an iterative process following the
preferential attachment rule, our PLN model may be explained by the following
process: {\em ``A node joins the network through an introduction node, and
  builds connections within the local community; then it completes its growth
  by joining and growing in multiple other communities.''}  To realize a
generative model based on this intuition, we need a stochastic process derived from the PLN
distribution that balances the Pareto and the Lognormal process to drive the
growth rate of each node.
 
\para{Dynamic Graph Snapshots.}
To understand the growth of social networks, we perform $30$ daily crawls of
the San Francisco Facebook network during the month of November $2009$.  We
use the same crawling strategy used to capture other Facebook
datasets~\cite{interaction}.  We use these daily snapshot graphs to help 
us understand the rate of network growth and how new edges form.

\para{Network Generative Models.}  Since our PLN model is itself a
combination of a Pareto distribution and a Log-normal distribution, our
search for an iterative process should naturally integrate two
theories. Starting with ``Preferential Attachment''~\cite{Vazquez2003}, which
produces a Power-law distribution, we add the ``Law of Proportional
Effects''~\cite{Gibrat1932}, where the growth of the degree of a node, at
discrete time $t+1$, is multiplicative and independent from its actual size
$X_t$. This contributes to a Lognormal distribution.  We study the
correlation of growth and current degree using our 30 snapshot graphs of the
San Francisco network.  The results in Figure~\ref{fig:evol_ratio} confirm
that the degree growth of different nodes is not influenced by their starting
degree, meaning that the growth is a constant independent of
node degrees.

\para{A Two-Phase Iterative Algorithm.}
We envision a two-phase iterative algorithm that integrates fundamental
properties from the ``Law of Proportional Effects''~\cite{Gibrat1932} and
``Preferential Attachment''~\cite{Vazquez2003}.  The bimodal connotation
is meant to highlight the ability of our algorithm to integrate
the Power-law and Lognormal probability models.

Our two-phase algorithm alternates between adding new nodes to the network
using a preferential attachment model, and growing the connectivity among
nodes using the law of proportionate effects.  Its two phases are driven by a
probability parameter $p$.  We do not claim to provide a formal derivation of
the algorithm or its proof of correctness.  Both are beyond the scope of this
paper.  Here we limit ourselves to preliminary results to validate our
underlying intuition.

\begin{figure}[t]
\centering
\epsfig{figure=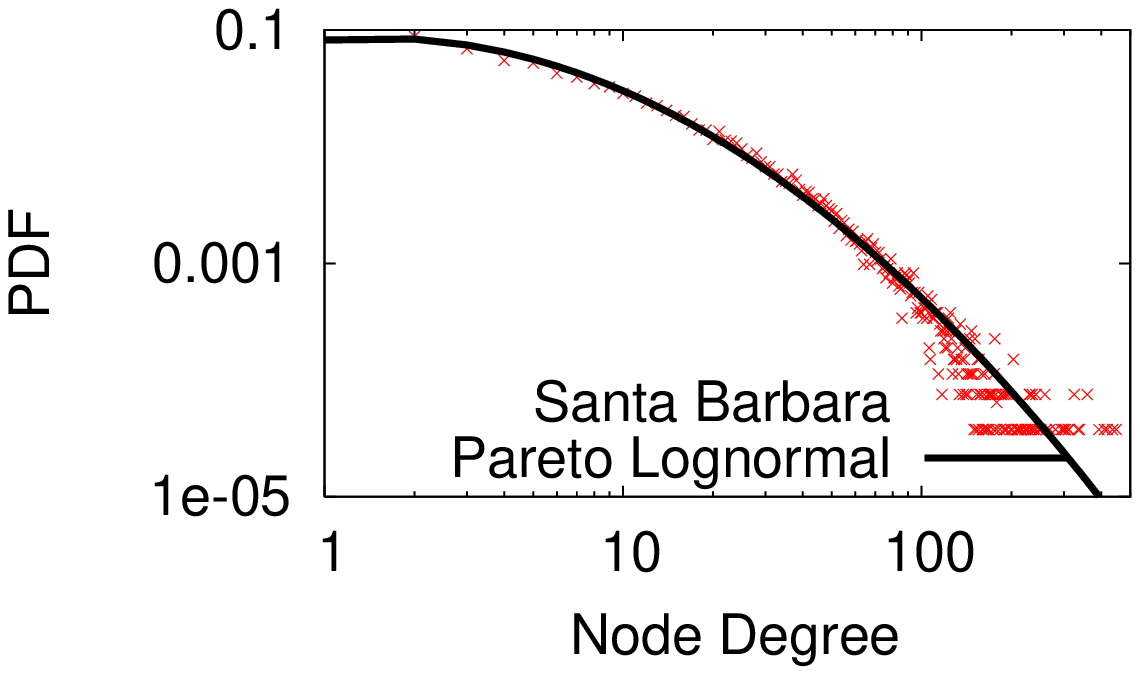, width=2.45in}
\epsfig{figure=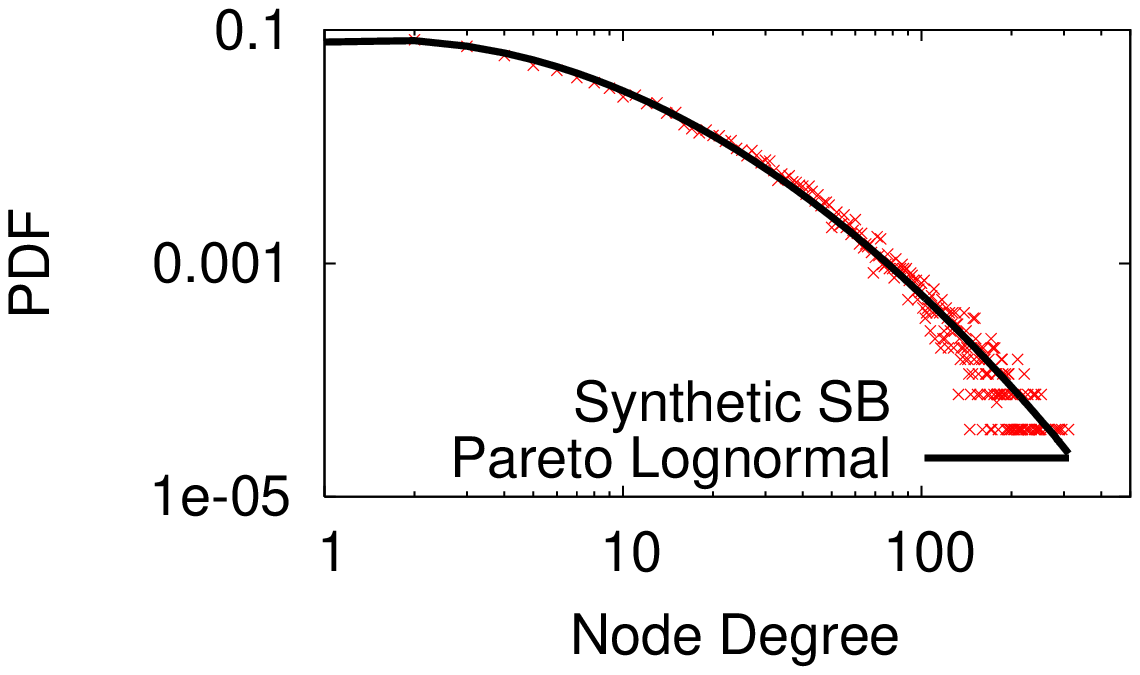, width=2.45in}
\caption{\small The comparison between the PDFs of the real and synthetic
  Santa Barbara network shows that our generative model provides an accurate
  reproduction of the original network.}
\label{fig:sythentic_real}
\end{figure}

\para{Initial Validation.} 
To evaluate our intuition about the generative method, we implement a simple
instance of a ``Two-Phased Iterative Algorithm'' summarized above, and run a
structural comparison between the generated graphs and the real social
network graphs.  We fit the parameters of this algorithm by following the
approach presented in~\cite{modeling_www}.

Our goal is to show that our intuitive algorithm for generating PLN
distributions can generate synthetic graphs that accurately reproduce the
degree distribution of our real data.  In Figure~\ref{fig:sythentic_real}, we
show the similarity between the pdf of node degrees from the real data,
versus those from the synthetic graph produced by our algorithm.  The results
are consistent across all seven of our datasets, and we only show results
only on Santa Barbara Network for brevity\footnote{These synthetic graphs
  also show similar structural characteristics to the original graphs,
  including clustering coefficient, Knn, separation metrics, and betweenness
  centrality.}.  Given these initial results, we are continuing to work on
formalizing the mathematical formulation of our generative model and a more
thorough analysis via larger, more complete synthetic datasets.

\section{Related Work}
\label{sec:related}
\para{Social Networks.} 
Historically, both online and offline social 
networks have been explained through the seminal Power-law model. 
Power-law is often described with the ``rich gets richer'' paradigm which
has been proven to hold in real datasets across multiple disciplines,
including Internet router topology graphs,
biological graphs~\cite{powerlawfit,powerlaw_bio}, human mobility
traces~\cite{human_mobility}, etc.   

One of the first OSN study was conducted on Club Nexus
website~\cite{Adamic2003}.  Later analytical studies attracted attention for their
large scale, including CyWorld, MySpace and Orkut~\cite{Ahn07},
YouTube, Flickr, LiveJournal in~\cite{socialnets-measurement}, and the most
recent studies of Facebook~\cite{interaction} and Twitter\cite{sue_twitter}.
More recently, researchers have begun to investigate the temporal properties
of OSNs~\cite{Kumar2006,flickr-growth,Leskovec2005}.

Preliminary analysis of OSN structures in these and other studies has shown
that the degree distribution does not follow a pure Power Law distribution. 
As a result, followup work proposed to segment these distributions and fit
the segmented pieces with distinct Power-law
settings~\cite{sampling_facebook,Ahn07}.

\para{Social Applications and Systems.}
We have shown that our proposed PLN is statistically more accurate in
describing OSNs than the seminal Power-law model.  We believe that many
social applications and protocols designed based on the Power-law assumption
need to be re-evaluated, especially algorithms and protocols that rely on the
population of high degree nodes or their connectivity. Examples include
distributed resource replication strategies to minimize routing delay and
social search, epidemics dissemination strategies to maximize information
spread~\cite{influence_kdd09}, landmark selection strategies to accurately
predicts shortest paths in graphs~\cite{social_landmark}, community detection
to improve social recommendation systems, and social attack
strategies~\cite{link_privacy}. 

\section{Conclusion}
\label{sec:conclusion}

Degree distributions are incredibly important tools for studying and
understanding the structure and formation of social networks. They give us 
insight into network structure, and are the foundations of generative models
that model growth and network
dynamics~\cite{pareto_cutoff,Leskovec2008,Leskovec2007}.  Finally,
they are key tools in the design and analysis of efficient algorithms for a
number of challenging graph problems.

Our work sheds light on an existing discrepancy between the commonly used
Power-law model and real measurement data from today's online social
networks.  While most prior studies have ignored this error, we show that it
is consistent across different communities and propose the Pareto-Lognormal
distribution (PLN) as a more accurate alternative.  Our analysis shows that
PLN significantly outperforms existing elementary models, and is the most
accurate and efficient of all complex distribution models we studied.
Finally, we analytically quantify the magnitude of error reduction achieved
by moving from the Power-law to the PLN model, and confirm our analysis 
with empirical data. 

\begin{small}

\end{small}

\end{document}